\documentclass[12pt]{article}
\usepackage[utf8]{inputenc}
\usepackage[titletoc,toc,title]{appendix}
\usepackage{bbding,comment}
\usepackage{multirow}
\usepackage{tabulary}
\usepackage{subfloat}
\usepackage{booktabs}
\usepackage{authblk}
\usepackage{amsmath}
\usepackage{amssymb}
\usepackage{amsthm}
\usepackage{graphicx}
\usepackage{hyperref}
\usepackage{caption}
\usepackage{subcaption}
\usepackage[normalem]{ulem}
\usepackage{longtable}
\usepackage{afterpage}
\usepackage{bm}
\usepackage{bbm}
\usepackage{lscape}
\useunder{\uline}{\ul}{}
\usepackage[table,xcdraw]{xcolor}
\usepackage{xcolor}
\usepackage{makecell}
\usepackage{multirow}
\usepackage{rotating}
\usepackage{tikz}
\usepackage{enumitem}
\usepackage{amsfonts}
\usepackage{tabularx}
\usepackage[section]{placeins}
\usepackage[ruled,linesnumbered]{algorithm2e}
\RestyleAlgo{ruled}
\setlist{nosep}

\newtheorem{theorem}{Theorem}[section]
\newtheorem{definition}[theorem]{Definition}
\newtheorem{proposition}[theorem]{Proposition}

\newtheorem{corollary}[theorem]{Corollary}
\newtheorem{lemma}[theorem]{Lemma}
\newtheorem{remark}[theorem]{Remark}
\newcolumntype{R}{>{\raggedleft\arraybackslash}X}

\renewcommand{\l}{\left}
\renewcommand{\r}{\right}
\renewcommand{\Pr}{\mathbb{P}}
\newcommand{\R}{\mathbb{R}}
\newcommand{\Q}{\mathbb{Q}}
\newcommand{\E}{\mathbb{E}}

\renewcommand{\tilde}{\widetilde}
\newcommand{\q}{\quad}
\newcommand{\One}{\mathbf{1}}

\definecolor{darkblue}{rgb}{0.1,0.1,0.9}
\definecolor{darkred}{rgb}{0.9,0.1,0.1}

\hypersetup{colorlinks, allcolors= blue}

\makeatletter

\newcommand{\Rmnum}[1]{\expandafter\@slowromancap\romannumeral #1@}

\title{\bf Pareto-Optimal Peer-to-Peer Risk Sharing\vspace{0.2cm}\\with Robust Distortion Risk Measures}

\author[$\scriptsize{\text{\CrossOpenShadow}}$]{Mario Ghossoub}
\affil[$\scriptsize{\text{\CrossOpenShadow}}$]{Department of Statistics and Actuarial Science, University of Waterloo.\protect\\ Email: mario.ghossoub@uwaterloo.ca
\thanks{Mario Ghossoub acknowledges financial support from the Natural Sciences and Engineering Research Council of Canada (NSERC Grant No.\ 2024-03744). Michael B.\ Zhu acknowledges financial support from the Society of Actuaries through the Hickman Scholars Program.}\vspace{0.35cm}}

\author[$\star$]{Michael B. Zhu}
\affil[$\star$]{Department of Statistics and Actuarial Science, University of Waterloo.\protect\\ Email: mbzhu@uwaterloo.ca\vspace{0.35cm}}

\author[$\ddagger$]{Wing Fung Chong}
\affil[$\ddagger$]{Maxwell Institute for Mathematical Sciences and Department of Actuarial Mathematics and Statistics, Heriot-Watt University. Email: alfred.chong@hw.ac.uk\vspace{0cm}}

\date{\today}

\begin{document}
\sloppy

\maketitle

\begin{abstract}
We study Pareto optimality in a decentralized peer-to-peer risk-sharing market where agents' preferences are represented by robust distortion risk measures that are not necessarily convex. We obtain a characterization of Pareto-optimal allocations of the aggregate risk in the market, and we show that the shape of the allocations depends primarily on each agent's assessment of the tail of the aggregate risk. We quantify the latter via an index of probabilistic risk aversion, and we illustrate our results using concrete examples of popular families of distortion functions. As an application of our results, we revisit the market for flood risk insurance in the United States. We present the decentralized risk sharing arrangement as an alternative to the current centralized market structure, and we characterize the optimal allocations in a numerical study with historical flood data. We conclude with an in-depth discussion of the advantages and disadvantages of a decentralized insurance scheme in this setting.
\end{abstract}

\newpage
\section{Introduction}

The idea that individuals act to maximize their heterogeneous preferences in a context of uncertainty and scarcity of resources is the very reason for the existence of exchange markets. A practically relevant example of such markets for actuaries and insurance professionals are the markets for sharing risks between agents with differing preferences, beliefs, and levels of risk aversion. Broadly speaking, markets for exchanging risks 
can be classified into two categories: centralized risk-sharing markets and decentralized risk-sharing markets. In centralized risk-sharing markets, a central agent (or collection thereof) acts as the supplier of insurance. This is the traditional model of insurance, of which the basic theoretical component is the notion of an insurance contract (a premium and an indemnity function). Decentralized risk-sharing markets consist of a pool of agents who wish to directly insure each other, without recourse to a central insurance provider. The agents need to determine \textit{a priori} a way to allocate the aggregate risk among the pool. Examples of these markets include peer-to-peer (P2P) insurance arrangements, which have gained significant popularity in recent years.

The study of Pareto-optimal insurance contracting within centralized insurance markets has its roots in the two-agent market consisting of one policyholder and one insurer. In early foundational work, Arrow \cite{arrow1971essays} and Borch \cite{borch1960attempt} show that at a Pareto optimum, the indemnity function in a two-agent market is full coverage above a constant deductible, provided that the insurer is risk-neutral and the policyholder's preference admits a risk-averse expected-utility representation. Numerous extensions of the classical model were subsequently proposed, mostly aiming at introducing more realistic models of policyholder behavior into the theory 
(see Gollier \cite{Gollier2013} or Schlesinger \cite{Schlesinger2000} for an overview). For instance, the effect of belief heterogeneity between the agents was examined by Ghossoub \cite{ghossoub2017arrow,Ghossoub2019a,Ghossoub2019b}, Boonen and Ghossoub \cite{BoonenGhossoub2019b,BoonenGhossoub2020a}, and Ghossoub et al.\ \cite{GJR2023}. In the literature on ambiguity in optimal insurance design, the majority of the literature considered ambiguity on the side of the insured, as in Bernard et al.\ \cite{Bernardetal2015} and Xu et al.\ \cite{Xuetal2018}. Several extensions were proposed in Amarante et al.\ \cite{AGP2015}, Ghossoub \cite{Ghossoub2019c,Ghossoub2019b}, and Birghila et al.\ \cite{BGB2023}, for instance. Insurance market models in which the insurer and/or policyholders have preferences represented by risk measures were also examined, such as Value-at-Risk (VaR) and Expected Shortfall (ES) (e.g., Cai et al.\ \cite{cai2008varcte}, Chi and Tang \cite{chi2013optimal}, Cheung et al.\ \cite{cheung_et_al_2015_convex_ordering}, and Asimit et al.\ \cite{asimit_et_al_2021_multiple_environments}), and distortion risk measures (DRMs) (e.g., Cui et al.\ \cite{cui2013optimal}, Assa \cite{assa2015optimal}, and Zhuang et al.\ \cite{Zhuangetal2016}). Extensions of these results to a centralized market with multiple suppliers of insurance have been examined by Boonen et al.\ \cite{Boonenetal2016}, Asimit and Boonen \cite{asimit2018gametheory}, Boonen and Ghossoub
\cite{BoonenGhossoub2019b,BoonenGhossoub2020c}, and Zhu et al.\ \cite{zhuetal2023spne}, for instance.

However, the focus on identifying Pareto-optimal allocations often neglects the impact of market forces in a given insurance market. In particular, the asymmetry of bargaining power between the insurer and the policyholder can invalidate certain Pareto-optimal allocations as likely outcomes of this market. Furthermore, the study of Pareto efficiency takes prices (or premia) in the market as given, and it does not address the question of how premia are determined at a market equilibrium. A popular equilibrium concept that takes into account the aforementioned asymmetry between the insurer and the policyholder is the Stackelberg equilibrium (or Bowley optimum), which applies to a monopolistic insurance market where the insurer and the policyholder move sequentially in an economic game, with the insurer as the leader and the policyholder as the follower. Specifically, the insurer is able to set the prices of insurance before the policyholder has a chance to purchase insurance. The contracts that are expected to emerge in markets with this leader-follower structure are Stackelberg equilibria, which were first examined in an insurance context by Chan and Gerber \cite{chan1985reinsurer} for agents with exponential utility functions. These results were further extended to the case of DRMs by Cheung et al.\ \cite{cheung2019bowley}. However, of particular economic relevance is the relationship between these Stackelberg equilibria and Pareto optimality of allocations. To this end, Boonen and Ghossoub \cite{boonen2022bowley} show that when the insurer uses a linear pricing rule and agents' preferences are represented by DRMs, Stackelberg equilibria represent only a subset of Pareto optimal allocations. Moreover, in every Stackelberg equilibrium, the insurer's first-mover advantage allows it to charge prices so high that the policyholder is left with no incentive to purchase insurance. In other words, Stackelberg equlibria do not induce a welfare gain to the policyholder. Recent results from Ghossoub and Zhu \cite{GhossoubZhu2024stackelberg} show that this situation occurs even in markets with multiple policyholders, and Zhu et al.\ \cite{zhuetal2023spne} show that introducing competition on the supply side of the insurance market alleviates this problem.

As an alternative to centralized markets of insurance, decentralized insurance markets offer the agents the possibility of sharing risks directly among themselves, thereby avoiding interaction with traditional insurance providers. This would circumvent the potential no-welfare-gain equilibria described above, as there is no longer any danger of an insurance provider exploiting an advantage in bargaining power. A primary concern of decentralized risk-sharing markets is the Pareto-efficiency of the allocations of the aggregate risk. The seminal work of Borch \cite{borch1962} and Wilson \cite{Wilson1968} showed that when agents have risk-averse expected-utility (EU) preferences, each individual's allocation of risk is a nondecreasing deterministic function of the aggregate risk in the market at a Pareto optimum. This, in turn, leads to a complete characterization of efficient allocations (through the so-called Borch rule) and shows that Pareto optima (PO) are in fact comonotone. The subsequent literature has then examined several extensions of the classical model beyond EU. Notably, the effect of ambiguity-sensitive agents on optimal allocations has attracted significant attention in the literature. For instance, Dana \cite{Dana2002, Dana2004} and De Castro and Chateauneuf \cite{decastro2011} consider economies with Maxmin Expected Utility  mutiple-prior preferences \`a la Gilboa and Schmeidler \cite{GilboaSchmeidler1989maxmin}. Chateauneuf et al.\ \cite{Chateauneufetal2000}, Dana \cite{Dana2004}, De Castro and Chateauneuf \cite{decastro2011}, and Beissner and Werner \cite{BeissnerWerner2023} study economies with non-probabilistic uncertainty as in the Choquet-Expected Utility model of Schmeidler \cite{schmeidler89}. The more general class of variational preferences of Maccheroni et al.\  \cite{Maccheronietal2006}) was examined by Dana and LeVan \cite{DanaLeVan2010} and Ravanelli and Svindland \cite{RavanelliSvindland2014}, for instance. In the context of risk management, Pareto-optimal allocations have been studied when preferences satisfy the properties of translation invariance, risk aversion, convexity, coherence, or a combination thereof: we refer to Jouini et al.\ \cite{JouiniSchachermayerTouzi2008}, Filipovic and Svindland \cite{filipovic2008optimal}, Mao and Wang \cite{mao2020risk}, Ghossoub and Zhu \cite{ghossoubzhu2024}, and references therein for an overview of the topic. Of particular recent interest is the risk-sharing problem for agents with quantile-based risk measures that are not necessarily convex, as studied by Embrechts et al.\ \cite{embrechts2018quantile,embrechts2020quantile}, Liu \cite{liu2020weighted}. In contrast with the literature on risk-averse preferences, optimal allocations can be counter-monotone rather than comonotone, as shown in Lauzier et al.\ \cite{lauzier2024negatively} and Ghossoub et al.\ \cite{ghossoubetal2024counter}.

In practice, decentralized markets of insurance are often found in the form of risk-sharing pooling arrangements, in which every agent pays a contribution fee to a pool in return for partial coverage of their monetary risk. For more details on real-world examples of peer-to-peer risk sharing, we refer to Abdikerimova and Feng \cite{abdikerimova2022peer}. Decentralized insurance is also prevalent in markets for catastrophe risk: we refer to Bollman and Wang \cite{bollmann2019international} for an overview, as well as Feng et al.\ \cite{feng2023peer} for the flood risk market with proportional insurance. The allocation of the \emph{ex post} risk to each agent is generally performed according to a predetermined risk-sharing rule. Examples of these rules previously studied in the literature include the conditional mean risk-sharing rule (e.g., Denuit and Dhaene \cite{denuit2012convex}), as well as rules based on actuarial fairness and Pareto optimality as in Feng et al.\ \cite{feng2023peer}. Our focus in this paper is on the Pareto optimality of allocations. For more on desirable properties of risk-sharing rules and recent developments, we refer to Denuit et al.\ \cite{denuit2022risk}.

In this paper, we study Pareto optimality in a decentralized peer-to-peer risk-sharing market where agents' preferences are represented by robust distortion risk measures that are not necessarily convex. In our setting, robustness is with respect to the distortion function used, not with respect to the probability measure on the space. This is in contrast to Bernard et al.\ \cite{bernard2023robust}, but it is in line with Wang and Xu \cite{wang2023preference}, who argue that these risk measures can be motivated by the ambiguity present in each agent's preference. This type of robustness also emerges naturally in the context of coherent risk measures (e.g., Dana \cite{dana2005representation}). Moreover, we allow for heterogeneity both in the baseline probability measures used by the agents and in the sets of distortion functions used by the agents. The set of feasible allocations is assumed to consist of comonotone allocations, that is, allocations that are $1$-Lipschitz functions of the aggregate risk in the market. This is related to the so-called \textit{no-sabotage} condition of Carlier and Dana \cite{CarlierDana2003b,CarlierDana2005a} typically assumed in centralized insurance markets, which guarantees that no agent has an incentive to misreport their actual realized loss \textit{ex post}.

Our main result (Theorem \ref{thm:po_drm_characterization}) provides a characterization of Pareto-optimal allocations of the aggregate risk in the market. In particular, we show that the shape of the allocations depends primarily on each agent's assessment of the tail of the aggregate risk. We quantify the latter via an index of probabilistic risk aversion. In a special case of our result where the set of distortion functions for each agent is in fact a singleton (Corollary \ref{cor:characterization_no_robust}), we recover the explicit form of optimal allocations under distortion risk measures previously attained by Liu \cite{liu2020weighted}. We then provide several illustrations of our results using concrete examples of popular families of distortion functions.

As an application of our results, we revisit the market for flood risk insurance in the United States through the lens of peer-to-peer insurance. We present the decentralized risk sharing arrangement as an alternative to the current centralized market structure, and we characterize the optimal allocations in a numerical study with historical flood data. Specifically, we revisit the numerical study of Boonen et al.\ \cite{BCG2024JRI}, which examines a centralized market with the federal government acting as the sole provider of insurance. However, as argued by Ghossoub and Zhu \cite{GhossoubZhu2024stackelberg}, this centralized structure provides an incentive for the central authority to charge extremely high prices at the expense of the policyholders. We therefore present the decentralized risk-sharing scheme as an alternative, and we characterize the optimal allocations thereof using historical data. We find that while the decentralized risk-sharing market avoids Stackelberg equilibria, the trade-off is that the average welfare gain is lower than that of the centralized market.

The remainder of this paper is structured as follows. Section \ref{sec:formulation} formulates the optimal allocation problem for decentralized markets. Section \ref{sec:explicit} provides the characterization of Pareto optima in this setting, as well as some general results for common families of distortion risk measures. We apply our results to the market for flood risk insurance in the United States in Section \ref{sec:flood_risk}, in which we also discuss the advantages and disadvantages of a decentralized market scheme. Section \ref{sec:conclusion} concludes.

\section{Problem Formulation}
\label{sec:formulation}

Let $\mathcal{X}$ be the set of bounded real-valued measurable functions on a given measurable space $\left(\Omega,\mathcal{F}\right)$, and let $\mathcal{X}_+$ denote its non-negative elements. There are $n\in\mathbb{N}$ agents wishing to share their endowed risks among themselves without any central authority involvements. For each $i \in \mathcal{N} :=\{1,\ldots,n\}$, let $X_i \in \mathcal{X}_+$ denote the non-negative loss of the $i$-th agent. We consider a one-period, risk-sharing economy, where all risks are realized at the end of the period. The risk-sharing mechanism in this market is as follows.

For each $i \in \mathcal{N}$, agent $i$ pays \textit{ex ante} the contribution amount $\pi_i\in\mathbb{R}$ to a pool. At the end of the period, the aggregate loss $S:=\sum_{i=1}^{n}X_i$ is covered by the aggregate amount $\sum_{i=1}^{n}\pi_i$ in the pool. The residual aggregate loss to be shared among the agents is thus given by $\tilde S := \sum_{i=1}^{n}X_i-\sum_{i=1}^{n}\pi_i$. If $\tilde S \geq 0$, the pool is subject to a residual aggregate loss; whereas if $\tilde S \leq 0$, the pool has a monetary surplus to be shared among the agents at the end of the period. Let $Y_i$ be the shared loss by the $i$-th agent, $i\in\mathcal{N}$, from the residual aggregate loss $\tilde S$. For each $i \in \mathcal{N}$, the end-of-period, post-transfer risk exposure of agent $i$ is thus given by $Y_i+\pi_i$.

\medskip

\begin{remark}
Note that this formulation slightly generalizes the usual formulation in the literature by incorporating deterministic contributions from the agents to the pool at the beginning of the period. This can be compared to the usual formulation by defining $Z_i:=Y_i+\pi_i$ to be the resulting allocation to agent $i$. Note that for each $i\in\mathcal{N}$, the allocation $Z_i$ is almost surely bounded by below.
\end{remark}

We assume that the market is a comonotone risk-sharing market, as previously studied by Boonen et al.\ \cite{boonen2021competitive}. This is a market in which the admissible contracts are only those that are comonotone with the aggregate initial risk. Under this restriction, the contracts all satisfy the \textit{no-sabotage} condition of Carlier and Dana \cite{CarlierDana2003b,CarlierDana2005a}, which guarantees that no agent has an incentive to misreport their actual realized loss. This is one possible justification for the existence of such a market, since preventing \emph{ex-post} moral hazard is in the best interest of every agent. An alternative justification arises from the situation where each agent's preference is monotone with respect to the convex order. In this case, it is well-known that Pareto optima are comonotone (e.g., Ghossoub and Zhu \cite{ghossoubzhu2024}).

The set of \textit{ex ante} admissible decision variables is therefore given by:
\begin{align*}
\mathcal{A}:=\left\{\left(\left\{Y_i\right\}_{i=1}^{n},\left\{\pi_i\right\}_{i=1}^{n}\right)\in\mathcal{X}_+^n\times\mathbb{R}^{n}:\sum_{i=1}^{n}(Y_i+\pi_i)=\sum_{i=1}^{n}X_i,\,\{Y_i\}_{i=1}^n\mbox{ is comonotone}\footnote{A random vector $\{Z_i\}_{i=1}^n$ is said to be \emph{comonotone} if $\left[Z_i(\omega_1)-Z_j(\omega_2)\right]\left[Z_i(\omega_1)-Z_j(\omega_2)\right]\ge0$, for all $\omega_1,\omega_2\in \Omega$ and $i,j\in\{1,\dots,n\}$.}\right\}.
\end{align*}

\medskip

\begin{remark}
\label{rmk:comonotone_rep}
Note that $\{Y_i+\pi_i\}_{i=1}^n$ is comonotone if and only if $\{Y_i\}_{i=1}^n$ is comonotone. Recall that $\sum_{i=1}^n(Y_i+\pi_i)=\sum_{i=1}^nX_i=S$ is the aggregate risk present in the insurance market. By a standard result (e.g., Denneberg \cite[Proposition 4.5]{denneberg1994non}), if $\{Y_i+\pi_i\}_{i=1}^n$ is comonotone, then there exist increasing 1-Lipschitz functions $f_i$ such that $Y_i+\pi_i=f_i(S)$.
\end{remark}

The preferences of each agent $i$, for $i\in\mathcal{N}$, are represented by given risk measures $\rho_i$. Consequently, each agent $i$ measures their pre-transfer risk exposure by $\rho_i\left(X_i\right)$, and their post-transfer risk exposure by $\rho_i\left(Y_i+\pi_i\right)$. We recall some properties of risk measures below.

\begin{definition}
    A risk measure $\varrho:\mathcal{X}\to\R$ is said to be:
    \begin{itemize}
        \item {Monotone} if $\varrho(Z_1)\le\varrho(Z_2)$, for all $Z_1,Z_2\in\mathcal{X}$ such that $Z_1\le Z_2$.
        \medskip
        \item {Translation invariant} if $\varrho(Z+c)=\varrho(Z)+c$, for all $Z\in\mathcal{X}$ and $c\in\R$.
        \medskip
        \item {Law-invariant} if for all $Z_1,Z_2\in\mathcal{X}$ with the same distribution under $\Pr$, we have $\varrho(Z_1)=\varrho(Z_2)$.
    \end{itemize}
\end{definition}

\noindent We will assume throughout that all risk measures are monotone.

\medskip

\begin{definition}
\label{defn:ir_po}
A contract $\left(\left\{Y^*_i\right\}_{i=1}^{n},\left\{\pi^*_i\right\}_{i=1}^{n}\right)\in\mathcal{A}$ is said to be:
\begin{itemize}
\item {\bf Individually Rational} (IR) if it incentivizes the parties to participate in the market, that is, 
\begin{equation*}
\rho_i\left(Y^*_i+\pi^*_i\right)\leq\rho_i\left(X_i\right),\quad \forall \, i \in \mathcal{N}.
\end{equation*}

\medskip

\item {\bf Pareto Optimal} (PO) if it is IR and there does not exist any other IR contract $\left(\left\{Y_i\right\}_{i=1}^{n},\left\{\pi_i\right\}_{i=1}^{n}\right)$ such that
\begin{equation*}
\rho_i\left(Y_i+\pi_i\right)\leq\rho_i\left(Y^*_i+\pi^*_i\right),\quad \forall \, i \in \mathcal{N},
\end{equation*}
with at least one strict inequality. We denote the set of all PO allocations by $\mathcal{PO}$.
\end{itemize}
\end{definition}

\subsection{Pareto Optimality in Peer-to-Peer Arrangements}
\label{subsec:pareto}

First, we show that all comonotone allocations are translations of a suitable non-decreasing function of the aggregate endowment. To this end, define the set $\mathcal{G}$ by the following.
\begin{equation*}
\mathcal{G}:=\left\{\left\{g_i\right\}_{i=1}^{n}\,\middle|\,g_i:\mathbb{R}_+\rightarrow\mathbb{R}_+\text{ non-decreasing Borel-measurable},\text{ and }\sum_{i=1}^{n}g_i\left(\cdot\right)=\text{Id}\right\}.
\end{equation*}
Note that if $\{g_i\}_{i=1}^n\in\mathcal{G}$, then each $g_i$ is 1-Lipschitz. Furthermore, since each $g_i$ is non-negative, we have $g_i(0)=0$ for all $i\in\mathcal{N}$.

\begin{lemma}
\label{lem:translation}
    Let $(\{Y_i\}_{i=1}^n,\{\pi_i\}_{i=1}^n)\in\mathcal{A}$. Then there exist functions $\{g_i\}_{i=1}^n\in\mathcal{G}$ and constants $\{c_i\}_{i=1}^n\in\R^n$ such that
    \[Y_i+\pi_i=g_i(S)+c_i\quad\forall i\in\mathcal{N}\,,\]
    and $\sum_{i=1}^nc_i=0$\,.
\end{lemma}
\begin{proof}
    Let $(\{Y_i\}_{i=1}^n,\{\pi_i\}_{i=1}^n)$ be a comonotone allocation. By Remark \ref{rmk:comonotone_rep}, for each $i\in\mathcal{N}$, there exists an increasing 1-Lipschitz function $f_i:\R\to\R$ such that $f_i(S)=Y_i+\pi_i$. Define the function $g_i$ by
\begin{align*}
    g_i: \R&\to\R\\
    x&\mapsto f_i(x)-f_i(0)\,.
\end{align*}
    Then $g_i(0)=0$, and $g_i$ is non-decreasing and 1-Lipschitz, which implies that $g_i$ is non-negative when its domain is restricted to $\R_+$. Since $f_i(S)$ is an allocation, we have
    \begin{equation}
    \label{eq:allocation_s}
        S=\sum_{i=1}^nf_i(S)=\sum_{i=1}^ng_i(S)+\sum_{i=1}^nf_i(0)
    \end{equation}
    Substituting $0$ for $S$ in \eqref{eq:allocation_s} gives
    \begin{align*}
        0&=\sum_{i=1}^ng_i(0)+\sum_{i=1}^nf_i(0)\\
        0&=\sum_{i=1}^nf_i(0)\,,\\
    \end{align*}
    which implies that
    \[\sum_{i=1}^ng_i(S)=S\,,\]
    so $\{g_i\}_{i=1}^n\in\mathcal{G}$. Substituting this into \eqref{eq:allocation_s} gives $\sum_{i=1}^nf_i(0)=0$. Hence, defining $c_i:=f_i(0)$ implies the result.
\end{proof}

In the case where every agent uses a translation invariant risk measure, we have a characterization of PO in terms of solutions of a minimization problem. Let $\mathcal{S}$ denote the set of all solutions to the problem
\begin{equation}
    \label{eq:weighted_infconv}
    \inf_{(\{Y_i\}_{i=1}^n,\{\pi_i\}_{i=1}^n)\in\mathcal{IR}} \ \sum_{i=1}^{n}\rho_i(Y_i+\pi_i)\,.
\end{equation}

\begin{proposition}
\label{prop:characterization}
If, for each $i\in\mathcal{N}$, $\rho_i$ is translation invariant, then $\mathcal{PO}=\mathcal{S}$.
\end{proposition}

\begin{proof}
Note that $\mathcal{S}\subseteq\mathcal{PO}$ can be easily proved by contradiction. To show the reverse inclusion, assume, by way of contradiction, that there exist $\left(\left\{Y^*_i\right\}_{i=1}^{n},\left\{\pi^*_i\right\}_{i=1}^{n}\right)\in\mathcal{PO}$ such that $\left(\left\{Y^*_i\right\}_{i=1}^{n},\left\{\pi^*_i\right\}_{i=1}^{n}\right)\notin\mathcal{S}$. Then, there exist $\left(\left\{\tilde{Y}_i\right\}_{i=1}^{n},\left\{\tilde{\pi}_i\right\}_{i=1}^{n}\right)\in\mathcal{IR}$ such that
\begin{equation}
\sum_{i=1}^{n}\rho_i\left(\tilde{Y}_i+\tilde{\pi}_i\right)<\sum_{i=1}^{n}\rho_i\left(Y^*_i+\pi^*_i\right).
\label{eq:smaller_sum}
\end{equation}

Define $\mathcal{N}_1,\mathcal{N}_2,\mathcal{N}_3\subseteq\mathcal{N}$ such that,
\begin{equation*}
    \rho_i\left(\tilde{Y}_i+\tilde{\pi}_i\right)>\rho_i\left(Y^*_i+\pi^*_i\right),\quad\forall i\in\mathcal{N}_1,
\end{equation*}
\begin{equation*}
    \rho_i\left(\tilde{Y}_i+\tilde{\pi}_i\right)=\rho_i\left(Y^*_i+\pi^*_i\right),\quad\forall i\in\mathcal{N}_2,
\end{equation*}
\begin{equation*}
    \rho_i\left(\tilde{Y}_i+\tilde{\pi}_i\right)<\rho_i\left(Y^*_i+\pi^*_i\right),\quad\forall i\in\mathcal{N}_3.
\end{equation*}
Note that $\mathcal{N}_1,\mathcal{N}_2,\mathcal{N}_3$ is a partition of $\mathcal{N}$; also, by \eqref{eq:smaller_sum}, $\mathcal{N}_3\neq\varnothing$.

By assumption, $\left(\left\{Y^*_i\right\}_{i=1}^{n},\left\{\pi^*_i\right\}_{i=1}^{n}\right)\in\mathcal{PO}$, which implies $\mathcal{N}_1\neq\varnothing$. Define, for $i\in\mathcal{N}_1$,
\begin{equation*}
    \varepsilon_i=\rho_i\left(\tilde{Y}_i+\tilde{\pi}_i\right)-\rho_i\left(Y^*_i+\pi^*_i\right)>0.
\end{equation*}
Then, by \eqref{eq:smaller_sum}, there exist $\left\{\varepsilon_i\right\}_{i\in\mathcal{N}_3}$ such that, (i) $\varepsilon_i\geq 0$, for $i\in\mathcal{N}_3$, (ii) $\rho_i\left(\tilde{Y}_i+\tilde{\pi}_i+\varepsilon_i\right)\leq\rho_i\left(Y^*_i+\pi^*_i\right)$, for $i\in\mathcal{N}_3$, with at least one strict inequality, and (iii) $\sum_{i\in\mathcal{N}_3}\varepsilon_i=\sum_{i\in\mathcal{N}_1}\varepsilon_i$. Define
\begin{equation*}
    \hat{\pi}_i=\tilde{\pi}_i-\varepsilon_i,\quad\forall i\in\mathcal{N}_1,
\end{equation*}
\begin{equation*}
    \hat{\pi}_i=\tilde{\pi}_i,\quad\forall i\in\mathcal{N}_2,
\end{equation*}
\begin{equation*}
    \hat{\pi}_i=\tilde{\pi}_i+\varepsilon_i,\quad\forall i\in\mathcal{N}_3.
\end{equation*}
Note that $\left(\left\{\tilde{Y}_i\right\}_{i=1}^{n},\left\{\hat{\pi}_i\right\}_{i=1}^{n}\right)\in\mathcal{A}$; indeed,
\begin{equation*}
    \sum_{i=1}^{n}\hat{\pi}_i=\sum_{i=1}^{n}\tilde{\pi}_i-\sum_{i\in\mathcal{N}_1}\varepsilon_i+\sum_{i\in\mathcal{N}_3}\varepsilon_i=\sum_{i=1}^{n}\tilde{\pi}_i=\sum_{i=1}^{n}X_i-\sum_{i=1}^{n}Y_i.
\end{equation*}
Moreover, $\left(\left\{\tilde{Y}_i\right\}_{i=1}^{n},\left\{\hat{\pi}_i\right\}_{i=1}^{n}\right)\in\mathcal{IR}$, since
\begin{align*}
    \rho_i\left(\tilde{Y}_i+\hat{\pi}_i\right)=&\;\rho_i\left(\tilde{Y}_i+\tilde{\pi}_i\right)-\varepsilon_i=\rho_i\left(\tilde{Y}_i+\tilde{\pi}_i\right)-\left(\rho_i\left(\tilde{Y}_i+\tilde{\pi}_i\right)-\rho_i\left(Y^*_i+\pi^*_i\right)\right)\\=&\;\rho_i\left(Y^*_i+\pi^*_i\right)\leq\rho_i\left(X_i\right),\quad\forall i\in\mathcal{N}_1,
\end{align*}
\begin{equation*}
    \rho_i\left(\tilde{Y}_i+\hat{\pi}_i\right)=\rho_i\left(\tilde{Y}_i+\tilde{\pi}_i\right)=\rho_i\left(Y^*_i+\pi^*_i\right)\leq\rho_i\left(X_i\right),\quad\forall i\in\mathcal{N}_2,
\end{equation*}
\begin{equation}
    \rho_i\left(\tilde{Y}_i+\hat{\pi}_i\right)=\rho_i\left(\tilde{Y}_i+\tilde{\pi}_i+\varepsilon_i\right)\leq\rho_i\left(Y^*_i+\pi^*_i\right)\leq\rho_i\left(X_i\right),\quad\forall i\in\mathcal{N}_3,
\label{eq:strict_inequality_N_3}
\end{equation}
in which \eqref{eq:strict_inequality_N_3} has at least one strict inequality, which implies $\left(\left\{Y^*_i\right\}_{i=1}^{n},\left\{\pi^*_i\right\}_{i=1}^{n}\right)\not\in\mathcal{PO}$, a contradiction.
\end{proof}

\section{Pareto Optima for Robust Distortion Risk Measures}
\label{sec:explicit}

In this section, we provide an explicit characterization of Pareto-optimal allocations when agents use robust distortion risk measures. These are defined below.

\begin{definition}
\label{defn:robust_drm}
    A risk measure $\varrho:\mathcal{X}\to\R$ is a {robust distortion risk measure} if there exists a set of distortion functions $\mathcal{T}$ and a probability measure $\Pr$ on $(\Omega,\mathcal{F})$ such that for all random variables $Z$ on the probability space $(\Omega,\mathcal{F},\Pr)$,
        \[
            \varrho(Z)=\sup_{T\in\mathcal{T}}\int Z\,dT\circ\Pr\,.
        \]
\end{definition}

Robustness in the context of Definition \ref{defn:robust_drm} refers to the consideration of different distortion functions by the decision maker, not different probability distributions on the underlying measurable space. Our definition is in line with that of Wang and Xu \cite{wang2023preference}, who argue that these risk measures can be motivated by the ambiguity present in each agent's preference. We note that the use of the term \emph{robust distortion risk measure} is not consistent throughout the literature. For instance, Bernard et al.\ \cite{bernard2023robust} uses this term to refer to the case where the robustness relates to uncertainty in the distribution of the risk. In contrast, Definition \ref{defn:robust_drm} assumes that the distribution of a given risk is fixed with respect to a reference probability measure.

\begin{lemma}
\label{lem:convex_closed_pointwise}
    For a robust distortion risk measure with respect to a set of distortion functions $\mathcal{T}$, we have
    \[
        \sup_{T\in\mathcal{T}}\int Z\,dT\circ\Pr=\sup_{T\in\overline{\mathrm{co}}(\mathcal{T})}\int Z\,dT\circ\Pr\,,
    \]
    where $\overline{\mathrm{co}}(\mathcal{T})$ is the closed convex hull of $\mathcal{T}$ with respect to pointwise convergence.
\end{lemma}

\begin{proof}
    Since the Choquet integral is linear in $T$, we have that
        \[
            \sup_{T\in\mathcal{T}}\int Z\,dT\circ\Pr=\sup_{T\in\mathrm{co}(\mathcal{T})}\int Z\,dT\circ\Pr\,,
        \]
    where $\mathrm{co}(\mathcal{T})$ denotes the convex hull of $\mathcal{T}$. It suffices to show that the Choquet integral is continuous with respect to pointwise convergence of $T$. To this end, let $\{T^{(k)}\}_{k=1}^\infty$ be a sequence of continuous functions converging pointwise to $T$. Then we have
        \begin{align*}
            \lim_{k\to\infty}\int Z\,dT_k\circ\Pr
            &=\lim_{k\to\infty}\int_0^\infty T_k(\Pr(Z>x))\,dx
        =\int_0^\infty\lim_{k\to\infty}T_k(\Pr(Z>x))\,dx\\
            &=\int_0^\infty T(\Pr(Z>x))\,dx
            =\int Z\,dT\circ\Pr\,,
        \end{align*}
    where we may apply the dominated convergence theorem since
        \[
            \int_0^\infty|T_k(\Pr(Z>x))|\,dx\le\int_0^\infty\One_{\{||Z||_\infty>x\}}\,dx=||Z||_{\infty}<\infty\,.
        \]
\end{proof}

We now assume that for each $i\in\mathcal{N}$, the risk measure $\rho_i$ is a robust distortion risk measure. We allow for each agent to have heterogeneous beliefs, represented by different probability measures $\Q_i$ on the same measurable space. Specifically, we assume that for each $i\in\mathcal{N}$, there exists a set of distortion functions $\mathcal{T}_i$ such that
    \begin{equation}
    \label{eq:robust_drm}
        \rho_i(Z)=\sup_{T_i\in\mathcal{T}_i}\int Z\,dT_i\circ\Q_i\,,
    \end{equation}
for all risks $Z\in\mathcal{X}$. By Lemma \ref{lem:convex_closed_pointwise}, we may assume without loss of generality that each $\mathcal{T}_i$ is convex and closed with respect to pointwise convergence.

Note that distortion risk measures are translation invariant. For translation-invariant risk measures, a characterization of Pareto optimal allocations is given by Proposition \ref{prop:characterization}. However, in the case of robust distortion risk measures, a more explicit characterization of the set $\mathcal{S}$ is possible, as shown by the following result.

\begin{theorem}
\label{thm:po_drm_characterization}
    Suppose that for each $i\in\mathcal{N}$, the risk measure $\rho_i$ is given by a robust distortion risk measure as in \eqref{eq:robust_drm}. Then the following hold:
    \begin{enumerate}[label=(\roman*)]
        \item There exists a solution to the problem
            \begin{equation}
            \label{eq:inf_problem}
                \sup_{\{T_i\}_{i=1}^n\,\in\,\prod_{i=1}^n\mathcal{T}_i}\,\int_0^\infty\min_{i\in\mathcal{N}}\,\{T_i(\Q_i(S>x))\}\,dx\,.
            \end{equation}
            
        \item A necessary condition for the allocation $(\{Y_i^*\}_{i=1}^n,\{\pi_i^*\}_{i=1}^n)$ to be Pareto optimal is that
                \[
                    Y_i^*+\pi^*_i=g_i^*(S)+c_i^*\,,
                \]
            where $\{c_i^*\}_{i=1}^n\in\R^n$ is chosen such that $\sum_{i=1}^nc_i^*=0$ and $\{g_i^*(S)+c_i^*\}_{i=1}^n\in\mathcal{IR}$, and $\{g_i^*\}_{i=1}^n\in\mathcal{G}$ can be written in terms of the integrals of suitable functions $h_i$. Specifically, for each $i\in\mathcal{N}$, we can write $g_i^*(x)=\int_0^xh_i(z)\,dz$, where each $h_i:\R_+\to[0,1]$ is a function such that for almost every $x\in\R_+$,
                \[
                    \sum_{i\in L_x}h_i(x)=1\mbox{ and }\sum_{i\in L_x^C}h_i(x)=0\,,
                \]
            where
                \begin{align*}
                    L_x&:=\left\{i\in\mathcal{N}:
                        T_i^*(\Q_i(S>x))=\min_{j\in\mathcal{N}}\{T_j^*(\Q_j(S>x))\}\right\}\,,\\
                    L_x^C&=\mathcal{N}\setminus L_x\,,
                \end{align*}
            and $(T_1^*,\ldots,T_n^*)$ is a solution to \eqref{eq:inf_problem}.
    \end{enumerate}
\end{theorem}

\allowdisplaybreaks
\begin{proof}
\begin{enumerate}[label=(\roman*)]
    \item Since $S$ is essentially bounded by $||S||_{\infty}<\infty$, we have
            \[
                \sup_{\{T_i\}_{i=1}^n\,\in\,\prod_{i=1}^n\mathcal{T}_i}\,\int_0^\infty\min_{i\in\mathcal{N}}\,\{T_i(\Q_i(S>x))\}\,dx:=V\le||S||_{\infty}<\infty\,.
            \]
        Let $\{(T_1^{(k)},\ldots,T_n^{(k)}\}_{k=1}^\infty$ be a sequence such that
            \[
                \int_0^\infty\min_{i\in\mathcal{N}}\,\{T_i(\Q_i(S>x))\}\,dx
                    \ge V-\frac{1}{k}\,.
            \]
        Then by repeatedly applying Helly's compactness theorem (e.g., \cite[pp.\ 165-166]{Doob94}), there exists a subsequence $\{(T_1^{(k_j)},\ldots,T_n^{(k_j)}\}_{j=1}^\infty$ that converges pointwise to a limit $(T_1^*,\ldots,T_n^*)\in\prod_{i=1}^n\mathcal{T}_i$. By the proof of Lemma \ref{lem:convex_closed_pointwise}, the objective of \eqref{eq:inf_problem} is continuous with respect to pointwise convergence of the distortion functions. Hence,
            \begin{align*}
                V&\ge\int_0^\infty\min_{i\in\mathcal{N}}\,\{T_i^*(\Q_i(S>x))\}\,dx
                    =\lim_{j\to\infty}\int_0^\infty\min_{i\in\mathcal{N}}\,\{T_i^{k_j}(\Q_i(S>x))\}\,dx\\
                    &\ge V-\lim_{j\to\infty}\frac{1}{k_j}=V\,,
            \end{align*}
        which implies that $(T_1^*,\ldots,T_n^*)$ solves \eqref{eq:inf_problem}.

    \medskip

    \item Let $(\{Y_i\}_{i=1}^n,\{\pi_i\}_{i=1}^n)\in\mathcal{A}$. By Lemma \ref{lem:translation}, this allocation can be written in terms of functions $\{g_i\}_{i=1}^n\in\mathcal{G}$ and constants $\{c_i\}_{i=1}^n\in\R^n$ where $\sum_{i=1}^nc_i=0$. Conversely, if $\{g_i\}_{i=1}^n\in\mathcal{G}$ and $\{c_i\}_{i=1}^n\in\R^n$, then $(\{g_i(S)\}_{i=1}^n,\{c_i\}_{i=1}^n)\in\mathcal{A}$. Therefore
    \begin{align*}
        \inf_{\l(\{Y_i\}_{i=1}^{n},\{\pi_i\}_{i=1}^n\r)\in\mathcal{IR}\cap\mathcal{A}}\l\{\sum_{i=1}^{n}\rho_i\left(Y_i+\pi_i\right)\r\}&=\inf_{\l(\{g_i\}_{i=1}^{n},\{c_i\}_{i=1}^n\r)\in\mathcal{IR}\cap(\mathcal{G}\times\R^n)}\l\{\sum_{i=1}^{n}\rho_i\left(g_i(S)+c_i\right)\r\}\\
        &=\inf_{\l(\{g_i\}_{i=1}^{n},\{c_i\}_{i=1}^n\r)\in\mathcal{IR}\cap(\mathcal{G}\times\R^n)}\l\{ \sum_{i=1}^{n}\rho_i\left(g_i(S)\right)\r\}\,,
    \end{align*}
    where we write $(\{g_i\}_{i=1}^n,\{c_i\}_{i=1}^n)\in\mathcal{IR}$ when the allocation $(\{g_i(S)\}_{i=1}^n,\{c_i\}_{i=1}^n)\in\mathcal{IR}$. 
    This problem is solved by $(\{g_i^*\}_{i=1}^n,\{c_i^*\}_{i=1}^n)\in\mathcal{G}\times\R^n$ if and only if $\{g_i^*\}_{i=1}^n$ solves
    \begin{equation}
        \label{eq:inf_conv_simplified}
        \inf_{\{g_i\}_{i=1}^n\in\mathcal{G}}\l\{\sum_{i=1}^n\rho_i(g_i(S))\r\}\,,
    \end{equation}
    and the constants $c_i^*$ are chosen such that $(\{g_i^*(S)\}_{i=1}^n,\{c_i^*\}_{i=1}^n)\in\mathcal{IR}$. It remains to show that solutions to \eqref{eq:inf_conv_simplified} are of the given form.

    \medskip

    Since each $\rho_i$ is a robust distortion risk measure, \eqref{eq:inf_conv_simplified} can be written as 
        \[
            \quad\inf_{\{g_i\}_{i=1}^n\in\mathcal{G}}\l\{\sum_{i=1}^n\sup_{T_i\in\mathcal{T}_i}\int g_i(S)\,dT_i\circ\Q_i\r\}
            =\inf_{\{g_i\}_{i=1}^n\in\mathcal{G}}\,\sup_{\{T_i\}_{i=1}^n\in\prod_{i=1}^n\mathcal{T}_i}\l\{\sum_{i=1}^n\int g_i(S)\,dT_i\circ\Q_i\r\}\,.
        \]

    We now verify that the minimax equality holds for this problem. To this end, let $\mathcal{C}([0,M])$ denote the set of continuous functions on $[0,M]$, which is a Banach space under the supremum norm. Let $\mathcal{D}:=\R^{[0,1]}$ denote the space of functions from $[0,1]\to\R$, which is a topological vector space under the topology of pointwise convergence. With the embeddings $\mathcal{G}\subseteq\mathcal{C}([0,M])^n$ and $T_i\subseteq\mathcal{D}$, the objective function above can be viewed as a function from $\mathcal{C}([0,M])^n\times\mathcal{D}^n$ to $\R$.

    \medskip
    
    Note that the sets $\mathcal{G}$ and $\mathcal{T}_i$ are closed and convex. Since $\mathcal{G}$ is a closed subset of a product of 1-Lipschitz functions on the interval $[0,M]$, it is compact by the Arzela-Ascoli Theorem \cite[IV.6.7]{Dunford}. Furthermore, the objective function is linear in both $g_i$ and $T_i$. Finally, the objective is continuous with respect to the topologies above, since the Choquet integral is continuous with respect to the supremum norm\footnote{In fact, it is Lipschitz continuous \cite[Proposition 4.11]{MarinacciMontrucchio}.}, and it is also sequentially continuous with respect to pointwise convergence of $T_i$ as shown by Lemma \ref{lem:convex_closed_pointwise}.

    \medskip
    
    Therefore by Sion's minimax theorem \cite[Theorem 2.132]{barbuprecupanu}, the minimax equality holds. Exchanging the order of the infimum and supremum yields the problem
        \begin{equation}
            \label{eq:maximin}
            \sup_{\{T_i\}_{i=1}^n\in\prod_{i=1}^n\mathcal{T}_i}\,\inf_{\{g_i\}_{i=1}^n\in\mathcal{G}}\l\{\sum_{i=1}^n\int g_i(S)\,dT_i\circ\Q_i\r\}\,.
        \end{equation}

    By standard results on minimax problems (e.g., \cite[Section 2.3]{barbuprecupanu}), a necessary condition for $\{g_i^*\}_{i=1}^n$ to solve \eqref{eq:inf_conv_simplified} is that for every $(T_1^*,\ldots,T_n^*)$ solving \eqref{eq:maximin}, we have
        \begin{align*}
            \sum_{i=1}^n\int g_i^*(S)\,dT_i^*\circ\Q_i=\inf_{\{g_i\}_{i=1}^n\in\mathcal{G}}\l\{\sum_{i=1}^n\int g_i(S)\,dT_i^*\circ\Q_i\r\}\,.
        \end{align*}

    To complete the proof, it suffices to show that for every fixed choice of distortion functions $(T_1,\ldots,T_n)$, the functions $\{g_i^*\}_{i=1}^n$ solve the problem
        \begin{equation}
        \label{eq:inner_inf}
            \inf_{\{g_i\}_{i=1}^n\in\mathcal{G}}\l\{\sum_{i=1}^n\int g_i(S)\,dT_i\circ\Q_i\r\}
        \end{equation}
    if and only if $\{g_i^*\}_{i=1}^n$ are of the given form, and that the value of problem \eqref{eq:inner_inf} is
        \[
            \int_0^\infty\min_{i\in\mathcal{N}}\,\{T_i(\Q_i(S>x))\}\,dx\,.
        \]

    To this end, we first check that $\{g_i^*\}_{i=1}^n\in\mathcal{G}$. For all $x\in\R_+$, we have
    \begin{align*}
        \sum_{i=1}^ng_i^*(x)
        =\sum_{i=1}^n\int_0^xh_i(z)\,dz
    =\int_0^x\sum_{i=1}^nh_i(z)\,dz
        =\int_0^x1\,dz
        =x\,.
    \end{align*}
    Furthermore, since the derivative of $g_i^*$ is non-negative, $g_i^*$ is increasing, and so $\{g_i^*\}_{i=1}^n\in\mathcal{G}$ and $g_i^*(S)$ is comonotone with $S$. Now let $\{\tilde{g}_i\}_{i=1}^n\in\mathcal{G}$. We have
    \begin{align*}
        \sum_{i=1}^n\rho_i(\tilde{g}_i(S))
        &=\sum_{i=1}^n\int_0^\infty T_i(\Q_i(S>x))\tilde{g}_i'(x)\,dx
        =\int_0^\infty\sum_{i=1}^nT_i(\Q_i(S>x))\tilde{g}_i'(x)\,dx\\
        &\ge\int_0^\infty\sum_{i=1}^n\min_{i=1,\ldots,n}\{T_i(\Q_i(S>x))\}\tilde{g}_i'(x)\,dx\stepcounter{equation}\tag{\theequation}\label{ineq}\\
        &=\int_0^\infty\min_{i=1,\ldots,n}\{T_i(\Q_i(S>x))\}\sum_{i=1}^n\tilde{g}_i'(x)\,dx\\
        &=\int_0^\infty\min_{i=1,\ldots,n}\{T_i(\Q_i(S>x))\}\,dx\\
        &=\int_0^\infty\min_{i=1,\ldots,n}\{T_i(\Q_i(S>x))\}\cdot1+\sum_{i\in L_x^C}T_i(\Q_i(S>x))\cdot0\,dx\\
        &=\int_0^\infty\min_{i=1,\ldots,n}\{T_i(\Q_i(S>x))\}\cdot\sum_{i\in L_x}h_i(x)+\sum_{i\in L_x^C}T_i(\Q_i(S>x))\cdot h_i(x)\,dx\\
        &=\int_0^\infty\sum_{i\in L_x}\min_{i=1,\ldots,n}\{T_i(\Q_i(S>x))\}\cdot h_i(x)+\sum_{i\in L_x^C}T_i(\Q_i(S>x))\cdot h_i(x)\,dx\\
        &=\int_0^\infty\sum_{i\in L_x}T_i(\Q_i(S>x))\cdot h_i(x)+\sum_{i\in L_x^C}T_i(\Q_i(S>x))\cdot h_i(x)\,dx\\
        &=\int_0^\infty\sum_{i=1}^nT_i(\Q_i(S>x))h_i(x)\,dx\\
        &=\sum_{i=1}^n\int_0^\infty T_i(\Q_i(S>x))h_i(x)\,dx
        =\sum_{i=1}^n\rho_i(g_i^*(S))\,.
    \end{align*}
    Therefore $\{g_i^*\}_{i=1}^n$ solves \eqref{eq:inf_conv_simplified}.

    \bigskip

    To show the converse, let $(\{Y_i\}_{i=1}^n,\{\pi_i\}_{i=1}^n)\in\mathcal{A}$. By Lemma \ref{lem:translation}, we may write
    \[Y_i+\pi_i=\tilde{g}_i(x)+c_i\,,\]
    for some $\{\tilde{g}_i\}_{i=1}^n\in\mathcal{G}$ and $\{c_i\}_{i=1}^n\in\R^n$ such that $\sum_{i=1}^nc_i=0$. Suppose that $\tilde{g}_i$ are not of the specified form. Namely, $\sum_{i\in L_x^C}\tilde{g}_i'(x)>0$ on a set $\mathcal{A}$ of positive measure. Then for every $x$ in $\mathcal{A}$, we have
    \begin{align*}
        \sum_{i=1}^nT_i(\Q_i(S>x))\tilde{g}_i'(x)&=\sum_{i\in L_x}T_i(\Q_i(S>x))\tilde{g}_i'(x)+\sum_{i\in L_x^C}T_i(\Q_i(S>x))\tilde{g}_i'(x)\\
        &=\min_{i=1,\ldots,n}\{T_i(\Q_i(S>x))\}\tilde{g}_i'(x)+\sum_{i\in L_x^C}T_i(\Q_i(S>x))\tilde{g}_i'(x)\\
        &>\min_{i=1,\ldots,n}\{T_i(\Q_i(S>x))\}\tilde{g}_i'(x)+\sum_{i\in L_x^C}\min_{i=1,\ldots,n}\{T_i(\Q_i(S>x))\}\tilde{g}_i'(x)\\
        &=\min_{i=1,\ldots,n}\{T_i(\Q_i(S>x))\}
        =\sum_{i=1}^n\min_{i=1,\ldots,n}\{T_i(\Q_i(S>x))\}\tilde{g}_i'(x)\,,
    \end{align*}
    where the strict inequality follows because $L_x^C$ is non-empty and $\tilde{g}_i'(x)$ are not all zero for $i\in L_x^C$. Taking the integral over the set $\mathcal{A}$ of positive measure gives
    \[\int_\mathcal{A}\sum_{i=1}^nT_i(\Q_i(S>x))\tilde{g}_i'(x)\,dx>\int_\mathcal{A}\sum_{i=1}^n\min_{i=1,\ldots,n}\{T_i(\Q_i(S>x))\}\tilde{g}_i'(x)\,dx\,.\]

\noindent Therefore in this case, the inequality \eqref{ineq} is strict, and so $\tilde{g}_i(x)$ does not solve \eqref{eq:inf_conv_simplified}.
\end{enumerate}
\end{proof}

In other words, Theorem \ref{thm:po_drm_characterization} provides a necessary condition for an allocation to be Pareto-optimal allocation. This result is similar to the characterization obtained by Ghossoub and Zhu \cite[Theorem 4.15]{ghossoubzhu2024}. However, a key difference is that the present result of Theorem \ref{thm:po_drm_characterization} does not require concavity of the distortion function (or equivalently, the convexity of the risk measure). The present result also holds under a setting of heterogeneous beliefs, as represented by the different probability measures $\Q_i$. These stronger conclusions can be obtained in the setting of this paper because the admissible allocations are constrained to be comonotone, which was not the case in Ghossoub and Zhu \cite{ghossoubzhu2024}.

In particular, Pareto-optimal allocations are obtained as translations of suitable functions of the aggregate risk. Part (ii) of Theorem \ref{thm:po_drm_characterization} decomposes the optimal risk allocation for agent $i$ into a constant $c_i^*$ and a risky portion $g_i^*(S)$, which is normalized so that $g_i^*(0)=0$. Additionally, for any fixed $\{g_i^*\}_{i=1}^n$ given by Theorem \ref{thm:po_drm_characterization}, it is always possible to choose $\{c_i^*\}_{i=1}^n\in\R^n$ such that $\{g_i^*(S)+c_i^*\}_{i=1}^n\in\mathcal{IR}$. One such choice is given by
    \begin{equation}
    \label{eq:ir_example}
        c_i^*:=\rho_i(X_i)-\rho(g_i^*(S))\,,\quad i\in\{1,\ldots,n-1\}\,,\quad\quad
            c_n^*:=-\sum_{i=1}^{n-1}c_i^*\,.
    \end{equation}
Then by translation invariance of each $\rho_i$, we have $\rho_i(X_i)-\rho(g_i^*(S)+c_i^*)=0$ for $i\in\{1,\ldots,n\}$. Furthermore,
    \begin{align*}
        \rho_n(X_n)-\rho_n(g_n^*(S)+c_n^*)&=\rho_n(X_n)-\rho_n(g_n^*(S))+\sum_{i=1}^{n-1}(\rho_i(X_i)-\rho(g_i^*(S)))\\
            &=\left(\sum_{i=1}^n\rho_i(X_i)\right)-\left(\sum_{i=1}^n\rho_i(g_i^*(S))\right)\ge0\,,
    \end{align*}
where the last inequality follows since $\{g_i^*(S)+c_i^*\}_{i=1}^n\in\mathcal{S}$ and the no-risk-sharing arrangement $\{X_i\}$ is individually rational.

More generally, the role of the constants $c_i^*$ is to determine the distribution of the aggregate welfare gain from P2P insurance among the agents in the market. We define the aggregate welfare gain as the difference between the aggregate valuation of the individuals' initial risks, given by $\sum_{i=1}^n\rho_i^*(X_i)$, and the aggregate valuation of the agents' post-transfer risks, given by $\sum_{i=1}^n\rho_i(g_i^*(S)+c_i^*)$:
$$W:=\sum_{i=1}^n\left(\rho_i(X_i)-\rho_i(g_i^*(S)+c_i^*)\right)=\sum_{i=1}^n\left(\rho_i(X_i)-\rho_i(g_i^*(S))\right).$$
Note that $\sum_{i=1}^n\rho_i(g_i^*(S)+c_i^*)$ is the value of problem \eqref{eq:weighted_infconv}, and therefore does not depend on the choice of the Pareto-optimal allocation $\{g_i^*(S)+c_i^*\}_{i=1}^n$. Hence, the aggregate welfare gain $W$ is well-defined, and $W\ge0$. Now let $\{w_i\}_{i=1}^n\in\R_n$ be such that $\sum_{i=1}^nw_i=W$ and $w_i\ge0$ for all $i\in\mathcal{N}$. Fix $\{g_i^*\}_{i=1}^n$ satisfying Theorem \ref{thm:po_drm_characterization}, and define
    \[
        c_i^*:=\rho_i(X_i)-\rho(g_i^*(S))-w_i\,,\quad \forall i\in\mathcal{N}\,.
    \]
By construction, we have $\sum_{i=1}^nc_i^*=\left(\sum_{i=1}^n\rho_i(X_i)\right)-\left(\sum_{i=1}^n\rho(g_i^*(S))\right)-W=0$. The allocation $\{g_i^*(S)+c_i^*\}$ is the risk-sharing arrangement that provides a welfare gain of $w_i$ to each agent $i$, in the sense that
    \[
        \rho_i(X_i)-\rho(g_i^*(S)+c_i^*)=w_i\ge0\,.
    \]
Hence, $\{g_i^*(S)+c_i^*\}\in\mathcal{IR}$. Note that the above example \eqref{eq:ir_example} corresponds to the arrangement that allocates the entire welfare gain to agent $n$: i.e., $(w_1,\ldots,w_n)=(0,0,\ldots,W)$.


\subsection{Common Parametric Families of Distortion Functions}

When each set of distortion functions $\mathcal{T}_i$ is a singleton, there is no robustness in the risk measure of each agent, and we recover the simpler setting of P2P risk-sharing with distortion risk measures. Suppose that $\mathcal{T}_i=\{T_i\}$ for all $i\in\mathcal{N}$. As a special case of Theorem \ref{thm:po_drm_characterization}, we obtain a closed-form solution for Pareto optima that is known in the literature (e.g., Liu \cite[Theorem 3.3]{liu2020weighted}). This result is restated below.

\begin{corollary}
\label{cor:characterization_no_robust}
    Suppose that for each $i\in\mathcal{N}$, the risk measure $\rho_i$ is given by a robust distortion risk measure as in \eqref{eq:robust_drm}. Then an allocation $(\{Y_i^*\}_{i=1}^n,\{\pi_i^*\}_{i=1}^n)$ is Pareto optimal if and only if
        \[
            Y_i^*+\pi^*_i=g_i^*(S)+c_i^*\,,
        \]
    where $\{c_i^*\}_{i=1}^n\in\R^n$ is chosen such that $\sum_{i=1}^nc_i^*=0$ and $\{g_i^*(S)+c_i^*\}_{i=1}^n\in\mathcal{IR}$, and $\{g_i^*\}_{i=1}^n\in\mathcal{G}$ can be written in terms of the integrals of suitable functions $h_i$. Specifically, for each $i\in\mathcal{N}$, we can write $g_i^*(x)=\int_0^xh_i(z)\,dz$, where each $h_i:\R_+\to[0,1]$ is a function such that for almost every $x\in\R_+$,
        \[
            \sum_{i\in L_x}h_i(x)=1\mbox{ and }\sum_{i\in L_x^C}h_i(x)=0\,,
        \]
    where
        \begin{align*}
            L_x&:=\left\{i\in\mathcal{N}:
                T_i(\Q_i(S>x))=\min_{j\in\mathcal{N}}\{T_j(\Q_j(S>x))\}\right\}\,,\\
            L_x^C&=\mathcal{N}\setminus L_x\,.
        \end{align*}
\end{corollary}

\begin{proof}
    We apply Theorem \ref{thm:po_drm_characterization} in the special case where for $i\in\mathcal{N}$, each $\mathcal{T}_i$ is given by the singleton set $\{T_i\}$. However, note that we obtain a stronger result by showing sufficiency of the characterization as well as necessity. The result follows from the same proof as that of Theorem \ref{thm:po_drm_characterization}, after applying the fact that problems \eqref{eq:inf_conv_simplified} and \eqref{eq:maximin} are identical problems in this special case.
\end{proof}

In this subsection, we examine the consequences of Corollary \ref{cor:characterization_no_robust} when agents' preferences are determined by common parametric families of probability distortion functions. First, suppose that each $T_i$ is an inverse S-shaped distortion function from a parametric family, and let $\alpha_i$ denote the value of the parameter that determines $T_i$. It can be verified that certain values of the parameter $\alpha_i$ are ``more S-shaped'', which indicates a larger deviation from risk neutrality. To facilitate an economically meaningful comparison, we provide the following definition, which can be interpreted as a proxy for the level of risk aversion for a distortion function.

\medskip

\begin{definition}[Probabilistic Risk Aversion Index]
    For a given distortion $T_i$ assumed to be twice differentiable, we define the Probabilistic Risk Aversion Index of $T_i$ to be
        \[
            PRA_i(t):=-\frac{T_i^{\prime\prime}(t)}{T_i^\prime(t)}\,.
        \]
\end{definition}

This index is similar to the index of ambiguity aversion proposed by Carlier and Dana \cite{CarlierDana2008a}. Indeed, if $T_2=g\circ T_1$ where $g$ is a concave function, then it can be verified that $PRA_2\ge PRA_1$. Hence, $T_2$ is more risk averse than $T_1$, which is consistent with the notion that the risk aversion of distortion risk measures is due to concavity of the distortion function (see Yaari \cite{yaari} for more on this interpretation). Alternatively, we may use an index of relative probabilistic risk aversion, as defined below.

\begin{definition}[Relative Probabilistic Risk Aversion Index]
    For a given distortion $T_i$ assumed to be twice differentiable, we define the {Relative Probabilistic Risk Aversion Index} of $T_i$ to be
    \[RPRA_i(t):=t\times PRA_i(t)=-\frac{t\,T_i^{\prime\prime}(t)}{T_i^\prime(t)}\,.\]
\end{definition}

One example of inverse S-shaped distortion functions is the parametric family of Prelec-1 distortion functions introduced by Prelec \cite{Prelec98}, who showed that these distortion functions are inverse-S shaped with a unique inflection point at $(e^{-1},e^{-1})$. These distortion functions are of the form
\[T_i(t)=\exp(-(-\ln(t))^{\alpha_i})\,,\]
where $\alpha_i\in(0,1)$.

In the following, suppose that each $\rho_i$ is a distortion risk measure with respect to a Prelec-1 distortion function with parameter $\alpha_i$. Let $j,k\in\mathcal{N}$ such that $\alpha_j=\max_{i\in\mathcal{N}}\alpha_i$ and $\alpha_k=\min_{i\in\mathcal{N}}\alpha_i$.

\begin{proposition}
\label{prop:pra_prelec}
    For $t\in(0,e^{-1})$, the probabilistic risk aversion index $PRA_i(t)$ is larger for smaller values of the parameter $\alpha_i$. In particular,
        \[PRA_k(t)\ge PRA_i(t)\ge PRA_j(t)\,.\]
\end{proposition}
\begin{proof}
    For a given parameter $\alpha_i\in(0,1)$, we have
    \begin{align*}
        \frac{d}{dt}\,\exp(-(-\ln(t))^{\alpha_i})&=-\dfrac{\alpha_i\mathrm{e}^{-\left(-\ln\left(t\right)\right)^{\alpha_i}}\cdot\left(-\ln\left(t\right)\right)^{\alpha_i}}{t\ln\left(t\right)}\\
        \frac{d^2}{dt^2} \, \exp(-(-\ln(t))^{\alpha_i})&=\dfrac{\alpha_i\mathrm{e}^{-\left(-\ln\left(t\right)\right)^{\alpha_i}}\cdot\left(-\ln\left(t\right)\right)^{\alpha_i}\left(\ln\left(t\right)+\alpha_i\cdot\left(-\ln\left(t\right)\right)^{\alpha_i}-\alpha_i+1\right)}{t^2\ln^2\left(t\right)}\\
        PRA_i(t)&=\frac{\ln(t)+\alpha_i(-\ln(t))^{\alpha_i}-\alpha_i+1}{t\ln(t)}\,.
    \end{align*}
    Differentiating this with respect to $\alpha_i$ gives
    \[\frac{d}{d\alpha_i}\,PRA_i(t)=\frac{(-\ln(t))^{\alpha_i}+\alpha_i\cdot\ln(-\ln(t))\cdot(-\ln(t))^{\alpha_i}-1}{t\ln(t)}\,.\]
    For $t\in(0,e^{-1})$, we have $-\ln(t)>1$. Therefore
    \begin{align*}
        (-\ln(t))^{\alpha_i}+\alpha_i\cdot\ln(-\ln(t))\cdot(-\ln(t))^{\alpha_i}-1&>1^{\alpha_i}+\alpha_i\cdot\ln(1)\cdot1^{\alpha_i}-1
        =0.
    \end{align*}
    Since $t\ln(t)<0$, this implies that 
    \[\frac{d}{d\alpha_i}\,PRA_i(t)<0\,,\]
    which yields the desired result.
\end{proof}

This result can be interpreted in the following manner. Recall that in the context of distortion risk measures, each distortion function $T_i$ is applied to the survival function $\Pr(S>x)$. Hence, the argument of the distortion $t\in(0,e^{-1})$ represents the probability of a tail event. Proposition \ref{prop:pra_prelec} then implies that agent $k$ is most averse to tail events, provided that the probability of this tail event is no more than $e^{-1}\approx0.3679$.

\medskip

\begin{proposition}
\label{prop:prelec}
    Suppose each $\rho_i$ is distortion risk measure given by a Prelec-1 distortion function with parameter $\alpha_i$, and let $\alpha_j=\max_{i\in\mathcal{N}}\alpha_i$ and $\alpha_k=\min_{i\in\mathcal{N}}\alpha_i$. Then a Pareto-optimal allocation $(\{Y_i^*\}_{i=1}^n,\{\pi_i^*\}_{i=1}^n)$ is given by
    \[Y_i^*+\pi_i^*=g_i^*(S)+c_i^*\,,\]
    where $\{c_i\}_{i=1}^n\in\R^n$ is chosen so that the allocation is $\mathcal{IR}$, and
    \begin{align*}
        g_k^*(x)&=\min\{x,d^*\}\,,\\
        g_j^*(x)&=\max\{x-d^*,0\}\,,\\
        g_i^*(x)&=0,\mbox{ for }i\in\mathcal{N}\setminus\{j,k\}\,.
    \end{align*}
    where $d^*=\mathrm{VaR}_{e^{-1}}(S)$.
\end{proposition}
\begin{proof}
    We first show that 
    \begin{align*}
        j&\in L_x\,,\mbox{ if }\,\Pr(S>x)\le e^{-1}\,,\\
        k&\in L_x\,,\mbox{ if }\,\Pr(S>x)>e^{-1}\,.
    \end{align*}
    Suppose first that $\Pr(S>x)\le e^{-1}$. Then we have
    \begin{align*}
        \Pr(S>x)&\le e^{-1}\\
        \ln(\Pr(S>x))&\le-1\\
        -\ln(\Pr(S>x))&\ge1\,.
    \end{align*}
    It follows that for all $i\in\mathcal{N}$,
    \begin{align*}
        (-\ln(\Pr(S>x)))^{\alpha_j}&\ge(-\ln(\Pr(S>x)))^{\alpha_i}\\
        \exp(-(-\ln(\Pr(S>x)))^{\alpha_j})&\le\exp(-(-\ln(\Pr(S>x)))^{\alpha_i})\,,
    \end{align*}
    which implies that $j\in L_x$. The second case is similar. Suppose that $\Pr(S>x)>e^{-1}$. Then we have
        \[-\ln(\Pr(S>x))<1\,,\]
    from which it is straightforward to conclude that $k\in L_x$.
    
    Therefore, defining
    \begin{align*}
        h_j(x)&:=\mathbbm{1}_{\{x\ge d^*\}}\,,\\
        h_k(x)&:=\mathbbm{1}_{\{x<d^*\}}\,,\\
        h_i(x)&:=0,\mbox{ for }i\in\mathcal{N}\setminus\{j,k\}\,,
    \end{align*}
    satisfies the conditions of Theorem \ref{thm:po_drm_characterization}. The result follows by taking $g_i^*(x)=\int_0^xh_i(z)\,dz$ for all $i\in\mathcal{N}$.
\end{proof}

Consequently, when all agents use a Prelec-1 distortion function, a possible PO allocation is a translate of the contract where agent $k$ provides full coverage of the aggregate risk up to a limit $d^*$, and the excess aggregate loss beyond the level $d^*$ is covered by agent $j$. By Proposition \ref{prop:pra_prelec}, agent $k$ exhibits the most risk aversion toward tail events. This is reflected by the form of Pareto optima in this case, since agent $k$ does not cover losses beyond the deductible. Instead, the tail risk is assumed by agent $j$, who exhibits the least risk aversion toward tail events.

Prelec \cite{Prelec98} also defines a two-parameter family of distortion functions (Prelec-2 distortion functions), by
    \[T_i(t)=\exp(-\beta_i(-\ln(t))^{\alpha_i})\,,\]
where $\alpha_i\in(0,1)$ and $\beta_i>0$. It can be verified that these distortion functions have an inflection point at $(e^{-1},e^{-\beta_i})$. In this case, we have a result similar to that of Proposition \ref{prop:pra_prelec}.

\begin{proposition}
\label{prop:pra_prelec_2}
    If each $\rho_i$ is a distortion risk measure given by a Prelec-2 distortion function with parameters $\alpha_i,\beta_i$, then we have the following:
    \begin{enumerate}[label=(\roman*)]
        \item For a fixed $\alpha_i$ and for all $t\in(0,1)$, the probabilistic risk aversion index $PRA_i(t)$ is larger for smaller values of the parameter $\beta_i$.
        \medskip
        \item For a fixed $\beta_i$ and for all $t\in(0,\varepsilon)$ where $\varepsilon>0$, the probabilistic risk aversion index $PRA_i(t)$ is larger for smaller values of the parameter $\alpha_i$.
    \end{enumerate}
\end{proposition}
\begin{proof}
    Similarly to the calculation in the proof of Proposition \ref{prop:pra_prelec}, we have
    \begin{align*}
        PRA_i(t)&=\frac{\ln(t)+\alpha_i\beta_i(-\ln(t))^{\alpha_i}-\alpha_i+1}{t\ln(t)}\,,\\
        \frac{d}{d\alpha_i}\,PRA_i(t)&=\frac{\beta_i(-\ln(t))^{\alpha_i}+\alpha_i\beta_i\cdot\ln(-\ln(t))\cdot(-\ln(t))^{\alpha_i}-1}{t\ln(t)}\,,\\
        \frac{d}{d\beta_i}\,PRA_i(t)&=\frac{\alpha_i(-\ln(t))^{\alpha_i}}{t\ln(t)}\,.
    \end{align*}
    We see that $\frac{d}{d\beta_i}\,PRA_i(t)<0$ for all $t\in(0,1)$, which shows (i). For (ii), if we take
    \[\varepsilon=\min\l\{e^{-1},\exp\l(-\l(\frac{1}{\alpha_i}\r)^{\l(\frac{1}{\beta_i}\r)}\r)\r\}\,,\]
    then for $t\in(0,\varepsilon)$, we have
    \begin{align*}
        -\ln(t)&>1\,,\\
        \beta_i(-\ln(t))^{\alpha_i}&>1\,.
    \end{align*}
    Therefore
    \begin{align*}
        \beta_i(-\ln(t))^{\alpha_i}+\alpha_i\beta_i\cdot\ln(-\ln(t))\cdot(-\ln(t))^{\alpha_i}-1&>1+0-1
        =0.
    \end{align*}
    Since $t\ln(t)<0$, this implies
    \[\frac{d}{d\alpha_i}\,PRA_i(t)<0\,,\]
    which yields the desired result.
\end{proof}

Hence, if there exists an agent $k$ such that $\alpha_k=\min_{i\in\mathcal{N}}\alpha_i$ and $\beta_k=\min_{i\in\mathcal{N}}\beta_i$, then this agent exhibits the most risk aversion to tail events. A similar result to that of Proposition \ref{prop:prelec} is given below.

\begin{proposition}
\label{prop:prelec2}
    Suppose each $\rho_i$ is distortion risk measure given by a Prelec-2 distortion function with parameters $\alpha_i$ and $\beta_i$.
    \begin{enumerate}[label=\roman*)]
        \item Suppose that  $\beta_1=\beta_2=\ldots=\beta_n:=\beta$. Let $\alpha_j=\max_{i\in\mathcal{N}}\alpha_i$ and $\alpha_k=\min_{i\in\mathcal{N}}\alpha_i$. Then a Pareto-optimal allocation $(\{Y_i^*\}_{i=1}^n,\{\pi_i^*\}_{i=1}^n)$ is given by
        \[Y_i^*+\pi_i^*=g_i^*(S)+c_i^*\,,\]
        where $\{c_i\}_{i=1}^n\in\R^n$ is chosen so that the allocation is $\mathcal{IR}$, and
        \begin{align*}
            g_k^*(x)&=\min\{x,d^*\}\,,\\
            g_j^*(x)&=\max\{x-d^*,0\}\,,\\
            g_i^*(x)&=0,\mbox{ for }i\in\mathcal{N}\setminus\{j,k\}\,.
        \end{align*}
        where $d^*=\mathrm{VaR}_{e^{-1}}(S)$.

        \medskip
   
        \item Suppose that $\alpha_1=\alpha_2=\ldots=\alpha_n:=\alpha$. Let $\beta_j=\max_{i\in\mathcal{N}}\beta_i$. Then a Pareto-optimal allocation $(\{Y_i^*\}_{i=1}^n,\{\pi_i^*\}_{i=1}^n)$ is given by
        \[Y_i^*+\pi_i^*=g_i^*(S)+c_i^*\,,\]
        where $\{c_i\}_{i=1}^n\in\R^n$ is chosen so that the allocation is $\mathcal{IR}$, and
        \begin{align*}
            g_j^*(x)&=x\,,\\
            g_i^*(x)&=0,\mbox{ for }i\in\mathcal{N}\setminus\{j\}\,.
        \end{align*}
    \end{enumerate}
\end{proposition}
\begin{proof}
    Similarly to the proof of Proposition \ref{prop:prelec}, we show that
    \begin{align*}
        j&\in L_x\,,\mbox{ if }\,\Pr(S>x)\le e^{-1}\,,\\
        k&\in L_x\,,\mbox{ if }\,\Pr(S>x)>e^{-1}\,.
    \end{align*}
    Suppose first that $\Pr(S>x)\le e^{-1}$. Then we have
    \begin{align*}
        \Pr(S>x)&\le e^{-1}\\
        \ln(\Pr(S>x))&\le-1\\
        -\ln(\Pr(S>x))&\ge1\,.
    \end{align*}
    It follows that for all $i\in\mathcal{N}$,
    \begin{align*}
        (-\ln(\Pr(S>x)))^{\alpha_j}&\ge(-\ln(\Pr(S>x)))^{\alpha_i}\\
        \beta(-\ln(\Pr(S>x)))^{\alpha_j}&\ge\beta(-\ln(\Pr(S>x)))^{\alpha_i}\\
        \exp(-(-\ln(\Pr(S>x)))^{\alpha_j})&\le\exp(-(-\ln(\Pr(S>x)))^{\alpha_i})\,,
    \end{align*}
    which implies that $j\in L_x$. The second case is similar. Suppose that $\Pr(S>x)>e^{-1}$. Then we have
        \[-\ln(\Pr(S>x))<1\,,\]
    from which it is straightforward to conclude that $k\in L_x$. Applying Theorem \ref{thm:po_drm_characterization} gives (i).

    For (ii), note that if $\beta_j\ge\beta_i$, then
        \[\exp(-\beta_j(-\ln(t))^{\alpha})\le\exp(-\beta_i(-\ln(t))^{\alpha})\,.\]
    Therefore $j\in L_x$ for all $x\in\R_+$. The result follows by applying Theorem \ref{thm:po_drm_characterization}.
\end{proof}

Part (i) of Proposition \ref{prop:prelec2} is a similar result to that of Proposition \ref{prop:prelec} in the case where every agent has the same second parameter. That is, the agent with the largest aversion to tail events provides coverage up to a deductible. Part (ii) states that if every agent uses the same first parameter $\alpha$, then the agent with the least aversion to tail events provides full insurance.

Similar results can be obtained for the family of S-shaped distortion functions first introduced by Kahneman and Tversky \cite{ptKT79,cptKT92}:
\[T_i(t)=\frac{t^{\gamma_i}}{\l(t^{\gamma_i}+(1-t)^{\gamma_i}\r)^{1/\gamma_i}}\,,\]
where $\gamma_i\in(0.279,1]$. Let $\gamma_j=\max_{i\in\mathcal{N}}\gamma_i$ and $\gamma_k=\min_{i\in\mathcal{N}}\gamma_i$. 

\bigskip

\begin{remark}
\label{rmk:pra_prelec}
    It can be numerically verified that for any $t\in(0,\varepsilon)$ for a sufficiently small $\varepsilon>0$, the probabilistic risk aversion index $PRA_i(t)$ is larger for smaller values of the parameter $\gamma_i$. The probabilistic risk aversion index is plotted in Figure \ref{fig:kt_pra}. In particular,
    \[PRA_k(t)\ge PRA_i(t)\ge PRA_j(t)\,.\]
    Again, since small values of $t$ correspond to tail events, we see that agent $k$ exhibits the most aversion to the risk of tail events, whereas agent $j$ exhibits the least aversion.
\end{remark}

It can be shown using numerical methods that there exists $t^*\in(0,1)$ for which
\begin{align*}
    j&\in L_x\,,\mbox{ if }\,\Pr(S>x)\le t^*\,,\\
    k&\in L_x\,,\mbox{ if }\,\Pr(S>x)>t^*\,.
\end{align*}

Distortion functions for different values of the parameter $\gamma$ are shown in Figure \ref{fig:kt_functions}. It is clear from the figure that the minimum of the distortion functions is obtained by either the agent with the largest parameter (agent $j$), or by the agent with the smallest parameter (agent $k$). Thus, by the same arguments as those of Proposition \ref{prop:prelec}, it follows that a PO allocation of a similar form exists. That is, consider the contract such that agent $k$ provides full coverage of the aggregate risk up to a limit $d^*$, and the excess aggregate loss beyond the level $d^*$ is covered by agent $j$. Then a suitable translate of this contract is PO.

\begin{figure}[hbtp!]
	\centering
	\begin{subfigure}[t]{0.5\textwidth}
    \centering
    \includegraphics[height=8cm]{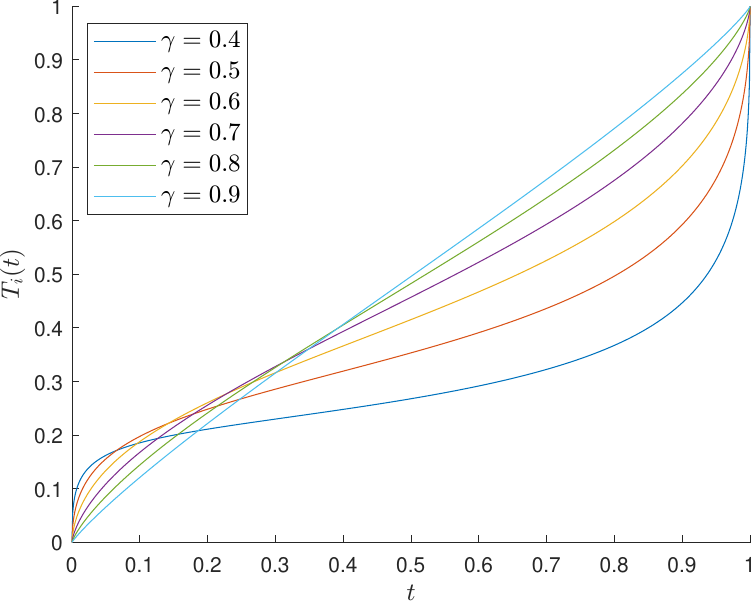}
	\caption{Distortion Functions}
    \label{fig:kt_functions}
	\end{subfigure}
	~
    \begin{subfigure}[t]{0.5\textwidth}
	\centering
	\includegraphics[height=8cm]{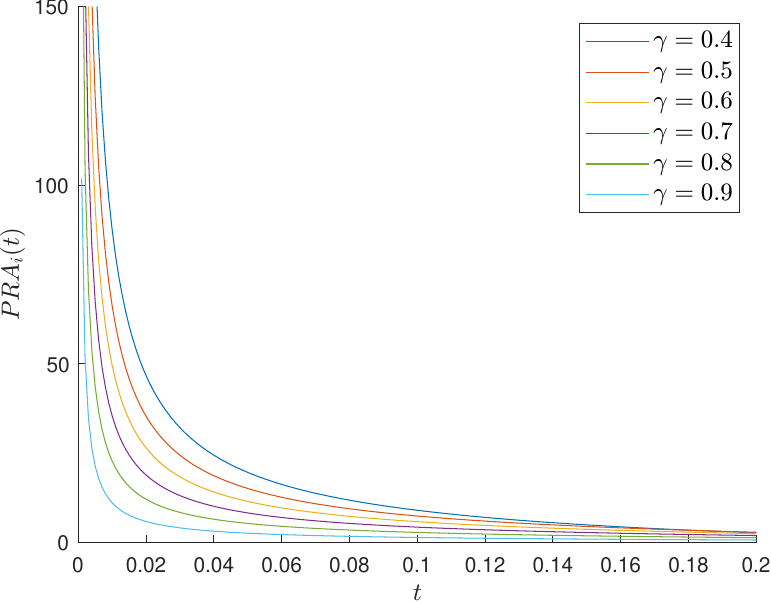}
    \caption{Probability Risk Aversion Index.}
	\label{fig:kt_pra}
    \end{subfigure}
    \caption{Kahneman-Tversky Inverse S-Shaped Distortion Functions}
\end{figure}


\section{Centralized Versus Decentralized Insurance for Flood Risk}
\label{sec:flood_risk}

As an application of our explicit characterization of Pareto-optimal risk-sharing contracts, we reexamine the problem of flood risk insurance in the United States. This setting has been examined in detail by Boonen et al.\ \cite{BCG2024JRI} and Ghossoub and Zhu \cite{GhossoubZhu2024stackelberg}, who consider a market where each State may insure their flood risk with the federal government. We refer to this structure as a \emph{centralized insurance market}, namely, one where insurance is only provided by a limited set of entities in the market. 

In contrast, the main result of the present paper characterizes Pareto optima in a \emph{decentralized insurance market} or \emph{peer-to-peer insurance market}. This market structure allows for agents to ensure each other's risk through participation in a risk sharing pool, which is not possible in a centralized insurance market. In the context of flood risk insurance, a decentralized insurance market is one where individual agents (communities, municipalities, or State-level agencies) agree to combine their flood risk exposures and each cover a portion of the aggregate risk exposure. For the purposes of this section, we assume that the individual agents are States that wish to pool State-level flood risk. This is mainly for ease of comparison with the centralized setting. However, the conclusions of our analysis are also applicable to agents at a higher level of granularity (e.g., individual communities, etc.). Feng et al.\ \cite{feng2023peer} also applied their peer-to-peer risk sharing results to the flood risk insurance in the United States, but they considered a market where only quota-share risk sharing rules (i.e., proportional insurance) are available.

\subsection{A Centralized Insurance Market}\label{numerical_centralized}

First, we revisit the scenario of a centralized insurance market, in which a single monopolistic insurer is the sole provider of insurance. In the context of flood risk in the United States, this insurer is the National Flood Insurance Program (NFIP)\footnote{A comprehensive overview of the NFIP can be found on the website of the Congressional Research Service: \url{https://sgp.fas.org/crs/homesec/R44593.pdf}.} overseen by the Federal Emergency Management Agency (FEMA). Under the NFIP, the federal government of the United States is ultimately responsible for financial gains and losses. 

A comprehensive numerical study of flood risk insurance in this centralized market is performed by Boonen et al.\ \cite{BCG2024JRI}, using a public dataset of claims statistics dating back to the year 1978. They find that by aggregating the flood risk across geographically distant regions, a centralized market is able to realize a significant welfare gain. In the following, we provide a similar but updated analysis of a centralized market. Firstly, the public NFIP dataset has recently been updated to include the amount of financial loss for each flood event, instead of only including the claim amount. As a result, we no longer need to use the claim amount as a proxy for the loss amount. Additionally, Ghossoub and Zhu \cite{GhossoubZhu2024stackelberg} are able to characterize Pareto optima without requiring that each agent uses a coherent risk measure. As a result, we are able to identify Pareto optimal contracts under more general conditions. Finally, in our example, the risk measure of the central agent will be given by the Expected Shortfall, which is the most popular coherent risk measure in practice. We recall the following standard definitions:

\begin{definition}
    The \emph{Value-at-Risk (VaR)} at level $\alpha\in(0,1)$ of a random variable $Z\in\mathcal{X}$ under the probability measure $\Pr$ is
    \[\mathrm{VaR}_\alpha^\Pr(Z):=\inf_{t\in\R}\l\{\Pr(Z>t)\le\alpha\r\}\,.\]
\end{definition}

\begin{definition}
    The \emph{Expected Shortfall (ES)} at level $\alpha\in(0,1)$ of a random variable $Z\in\mathcal{X}$ under the probability measure $\Pr$ is
    \[\mathrm{ES}_\alpha^\Pr(Z):=\frac{1}{\alpha}\int_0^\alpha\mathrm{VaR}_u^\Pr(Z)\,du\,.\]
\end{definition}

We will assume throughout that the distribution of the financial losses arising from floods in a given month adheres to the historical data. Specifically, we assume a discrete probability space $(\Omega,\mathcal{F},\Pr)$, where $\Omega$ has $553$ states, corresponding to each of the months from January 1978 to January 2024. The measure $\Pr$ assigns an equal probability to each of the states in $\Omega$. The risk measures used by each agent are distortion risk measures on this probability space. That is, for all $Z\in\mathcal{X}$ and $i\in\{1,\ldots,n\}$,
\[\rho_i(Z)=\int Z\,dT_i\circ\Pr\,,\]
where each $T_i$ is a distortion function. Let $\nu_i := T_i\circ\Pr$.

In the context of centralized insurance markets with $n$ agents and a central insurer, an insurance contract is a pair $\left(\{I_i\}_{i=1}^n,\{\pi_i\}_{i=1}^n\right)\in\mathcal{I}^n\times\R^n$, where $\mathcal{I}$ denotes an admissible set of indemnity functions (usually monotone and 1-Lipschitz). Here, $I_i$ represents the amount of coverage (i.e.\ the indemnity) that agent $i$ receives from the central insurer, in exchange for a premium of the amount $\pi_i$. In other words, agent $i$ cedes the risk $I_i(X_i)$ to the central insurer. We denote the retained risk of agent $i$ by $R_i(X_i):=X_i-I_i(X_i)$.

The concepts of individual rationality and Pareto optimality for centralized markets are defined similarly to IR and PO in decentralized markets (see Definition \ref{defn:ir_po}). In particular, IR allocations for centralized markets also incentivises the central insurer to participate. A characterization of Pareto optimal contracts in a centralized market is provided by the following result.

\begin{proposition}[Pareto Optima in a Centralized Market]
\label{prop:po_centralized}
Suppose that all agents use distortion risk measures, and that the insurer's risk measure is given by an Expected Shortfall at level $\alpha\in(0,1)$. Then the contract $\left(\{I_i^*\}_{i=1}^n,\{\pi_i^*\}_{i=1}^n\right)$ is Pareto-optimal if and only if:
\medskip
\begin{enumerate}
    \item There exist a probability measure $\Q^*$ and $[0,1]$-valued measurable functions $h_i$ such that for almost all $t\in(0,M)$,
            \begin{equation}
                \label{mif_form_v2}
                (I_i^*)'(t)=\begin{cases}
                        1,&\mbox{ if}\q\Q^*(X_i>t)<\nu_i(X_i>t)\,,\\
                        h_i(t),&\mbox{ if}\q\Q^*(X_i>t)=\nu_i(X_i>t)\,,\\
                        0,&\mbox{ if}\q\Q^*(X_i>t)>\nu_i(X_i>t)\,,\\
                    \end{cases}
            \end{equation}
        and $\Q^*$ solves
            \begin{equation}
    \label{measure_max_problem_v2}   \max_{\Q\in\mathcal{Q}} \, \sum_{i=1}^n\int_0^M\min\l\{\Q(X_i>t),\,\nu_i(X_i>t)\r\}\,dt\,,
            \end{equation}
        where
            \[
                \mathcal{Q}:=\l\{\Q\mbox{ a probability measure}: \Q<<\Pr,\,\frac{d\Q}{d\Pr}\le\frac{1}{\alpha}\r\}\,.
            \]

            \medskip
            
    \item $\{\pi_i^*\}_{i=1}^n$ are chosen so that $\left(\{I_i^*\}_{i=1}^n,\{\pi_i^*\}_{i=1}^n\right)$ is individually rational.
    \end{enumerate}
\end{proposition}

\begin{proof}
    It is well-known (e.g., McNeil et al.\ \cite[Theorem 8.14]{mcneil2015quantitative}) that the Expected Shortfall satisfies
        \[
            \mathrm{ES}_\alpha^\Pr(Z)=\max\l\{\E^\Q[Z]:\Q<<\Pr,\,\frac{d\Q}{d\Pr}\le\frac{1}{\alpha}\r\}\,,
        \]
    for all $Z\in\mathcal{X}$. The result then follows from Ghossoub and Zhu \cite[Corollary 3.7]{GhossoubZhu2024stackelberg}.
\end{proof}

We now apply the result of Proposition \ref{prop:po_centralized} to obtain Pareto-optimal contracts in the centralized insurance market. We consider a market with three agents: the States of California ($i=1$), New York ($i=2$), and Texas ($i=3$). The correlation between the monthly losses of these States is given by the matrix
    \[
        \begin{bmatrix}
            1&-0.0094&-0.0109\\
            -0.0094&1&-0.0044\\
            -0.0109&-0.0044&1
        \end{bmatrix}\,,
    \]
which indicates very little correlation. Some summary statistics for the monthly losses for these States are shown in Table \ref{table:summary}, in which all values are dollar amounts. We see that the financial losses are right-skewed, and that it is possible to experience losses far larger than the average.

\begin{table}[htbp!]
    \centering
    \begin{tabular}{cccc}
        \hline
        &California&New York&Texas\\
        \hline
        \hline
        Average&$1.6038\times10^6$&$1.1005\times10^7$&$3.3259\times10^7$\\
        Median&$2.7647\times10^4$&$1.4873\times10^5$&$4.7270\times10^5$\\
        $\mathrm{VaR}_{5\%}$&$4.9501\times10^6$&$7.1506\times10^6$&$4.6885\times10^7$\\
        Maximum&$1.1725\times10^8$&$4.3537\times10^9$&$9.4550\times10^9$\\
        Standard Deviation&$8.6640\times10^6$&$1.8625\times10^8$&$4.2063\times10^8$\\
        \hline
    \end{tabular}
    \caption{Summary Statistics for Monthly Losses Due to Floods.\vspace{0.4cm}}
    \label{table:summary}
\end{table}

We assume that the distortion functions $T_1$ and $T_2$ are both inverse S-shaped distortion functions, as in Kahneman and Tversky \cite{ptKT79,cptKT92}. Specifically, for $i=1,2$, we take
    \[
        T_i(t)=\frac{t^{\gamma_i}}{(t^{\gamma_i}+(1-t)^{\gamma_i})^\frac{1}{\gamma_i}}\,,
    \]
with $\gamma_1=0.4$ and $\gamma_2=0.5$. A parameter value of approximately $0.5$ has recently been estimated by Rieger et al.\ \cite{Riegeretal2017} to be descriptive of actual behaviour. We assume a power distortion for the remaining agent. Specifically,
    \[
        T_3(t)=t^{\gamma_3}\,,
    \]
with $\gamma_3=0.4$. Finally, we assume that the central insurer uses the Expected Shortfall with parameter $\alpha=15\%$. While this parameter may seem large, this is to allow for the central insurer's risk measure to be display enough risk tolerance that it will accept a significant portion of risk from the agents in the market. If this parameter is too low, then the central insurer would be too conservative and avoid taking on any flood risk, leading to a market where no insurance is possible. Indeed, it can be verified that under these parameters, no insurance is provided if $\alpha\approx2.5\%$.

The retained risk of each agent in a Pareto-optimal contract is shown in Figure \ref{fig:centralized1}.
Note that in this figure and subsequent figures, the amount of the premium is not shown, for visual clarity and to facilitate easier comparison between plots. In other words, each plot has been normalized so that an experienced loss of $\$0$ corresponds to a retention of $\$0$ as well. We see that in each State prefers to cede a portion of their tail risk to the central insurer, while choosing to retain some of the smaller financial losses. This suggests that a deductible contract is best in this situation, with losses beyond the deductible fully covered by the insurer. The value of this deductible is approximately $1.9047\times10^7$ for California, $4.0645\times10^7$ for New York, and $5.1287\times10^7$ for Texas.
 
\begin{figure}[htbp!]
	\centering
	\includegraphics[width=\textwidth]{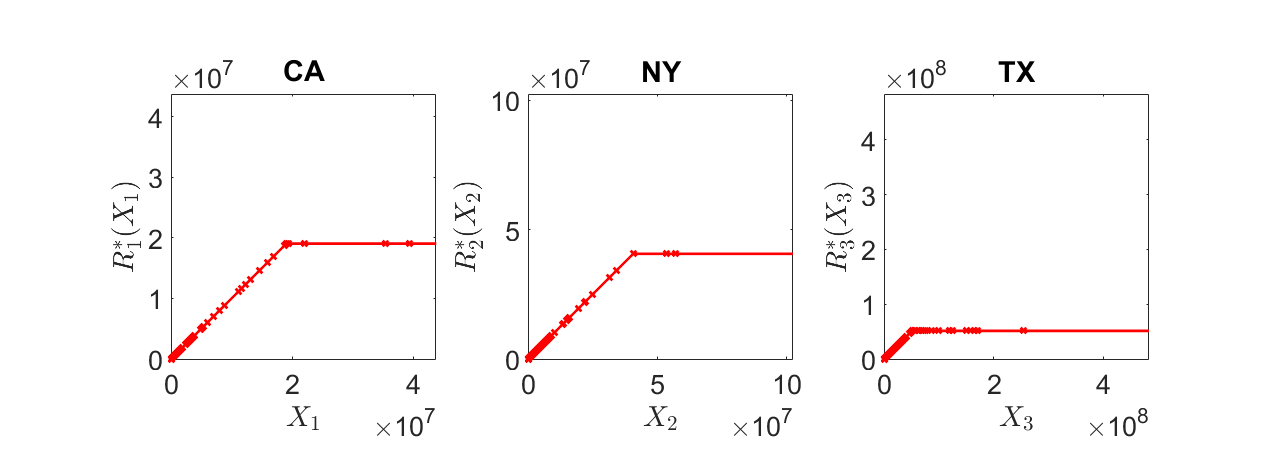}
	\caption{Centralized Pareto-Optimal Retention, CA/NY/TX
        }
    \label{fig:centralized1}
\end{figure}

As a measure of the effectiveness of this contract, we calculate the aggregate welfare gain that results from this arrangement. The welfare gain for each agent $i$ is the difference between their evaluation of the initial risk and their evaluation of the post-contract position:
    \[
        \rho_i(X_i)-\rho_i(R_i(X_i)+\pi_i)\,.
    \]
On the other hand, the welfare gain for the central insurer is the sum of premia received less their evaluation of their assumed risk:
    \[
        \sum_{i=1}^n\pi_i-\mathrm{ES}_\alpha^\Pr\left(\sum_{i=1}^nI_i(X_i)\right)\,.
    \]
The aggregate welfare gain is defined to be the sum of the welfare gain of each agent in the market, including the central insurer. Under our assumptions, the aggregate welfare gain is maximized in a Pareto-optimal allocation (see, e.g., Boonen et al.\ \cite[Theorem 2.1]{BCG2024JRI}). In the example above, the maximum aggregate welfare gain achievable through centralized insurance is $7.9785\times10^8$. The average gain per agent is in this example is $1.9946\times10^8$. Note that including the centralized insurer, there are four agents in the  market in this example.
    
However, while a significant welfare gain is possible in a centralized market, Ghossoub and Zhu \cite{GhossoubZhu2024stackelberg} warn that a monopolistic insurer with a first-mover's advantage can absorb the entirety of the welfare gain by increasing premium prices. 
In particular, the central insurer has an incentive to increase prices and ultimately leave each policyholder with no welfare gain at an equilibrium. The amount of premium paid by each State in a Stackelberg equilibrium in our example is shown in Table \ref{table:premia_stackelberg}. These amounts are extremely high, even exceeding the Value-at-Risk at the $5\%$ level. In this case, the entire aggregate welfare gain of $7.9785\times10^8$ is allocated only to the central insurer.

\begin{table}[htbp!]
    \centering
    \begin{tabular}{cc}
        \hline
        California&$9.1543\times10^6$\\
        New York&$1.7795\times10^8$\\
        Texas&$8.6323\times10^8$\\
        \hline
    \end{tabular}
    \caption{Premium Paid in a Stackelberg Equilibrium, CA/NY/TX.\vspace{0.4cm}}
    \label{table:premia_stackelberg}
\end{table}

\subsection{Peer-to-Peer Risk Sharing}\label{numerical_p2p}

As a partial solution to the concerns arising from the Stackelberg setting, the present paper proposes a peer-to-peer risk sharing scheme among the agents themselves without the presence of a central agency. In the context of flood risk, this is the scenario where the individual State governments of California, New York, and Texas choose to pool their flood risk and each assume a portion of the aggregate financial loss.

Applying the result of Theorem \ref{thm:po_drm_characterization} to the present example yields the Pareto-optimal allocation shown in Figure \ref{fig:p2p1}. The States of California and New York end up with all of the variability of the risk, since the risk measure of Texas is more risk averse under our chosen parameters. However, this does not mean that Texas is receiving free insurance, since each State must pay a premium for participation in the risk-sharing pool. Both Figure \ref{fig:p2p1} and Figure \ref{fig:centralized1} show that Texas prefers retaining less risk than the other agents.

\begin{figure}[htbp!]
	\centering	\includegraphics[width=\textwidth]{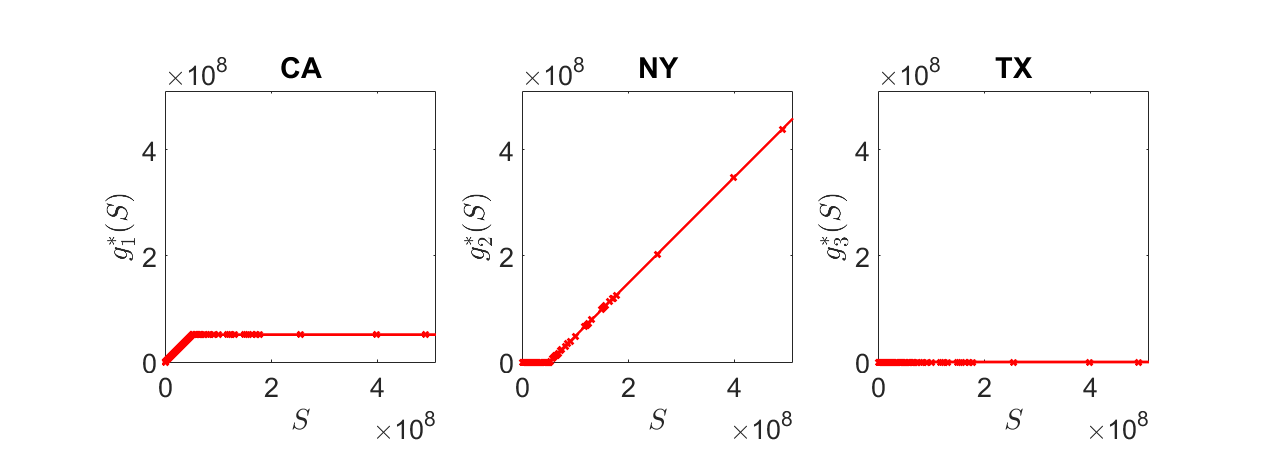}
	\caption{Decentralized Pareto-Optimal Distribution, CA/NY/TX.}
    \label{fig:p2p1}
\end{figure}

\begin{figure}[htbp!]
	\centering
	\includegraphics[width=0.5\textwidth]{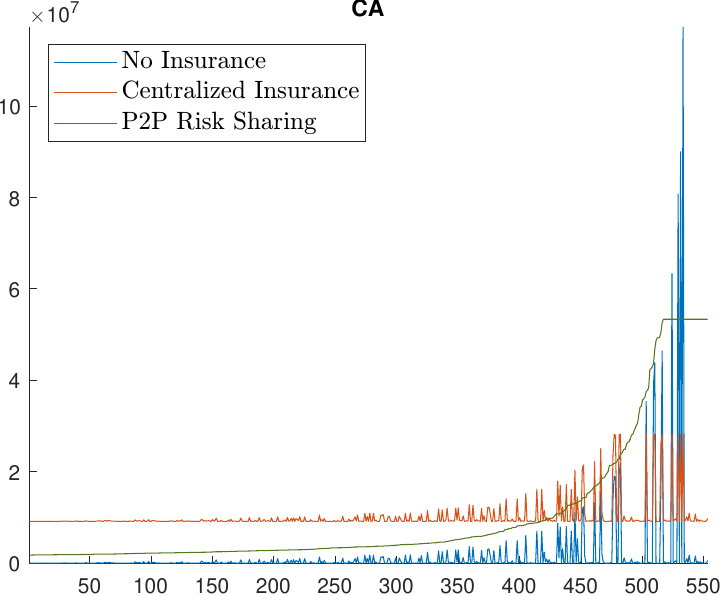}\\
        \medskip
	\includegraphics[width=0.5\textwidth]{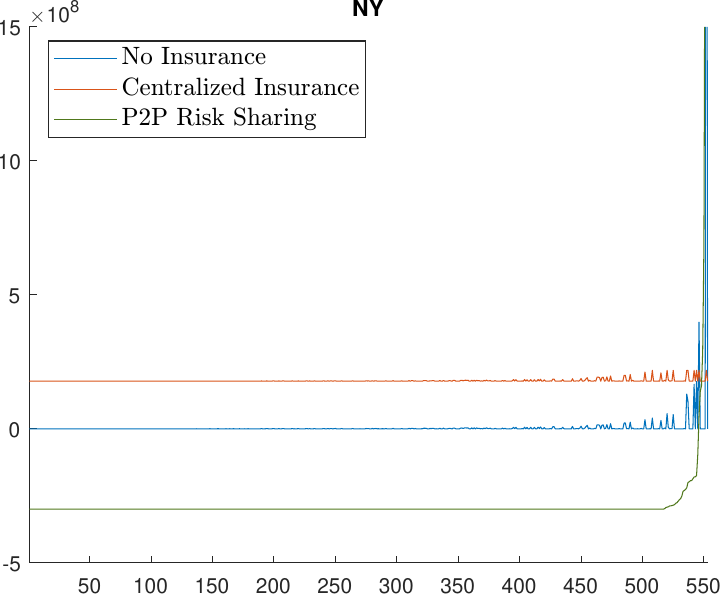}\\
        \medskip
	\includegraphics[width=0.5\textwidth]{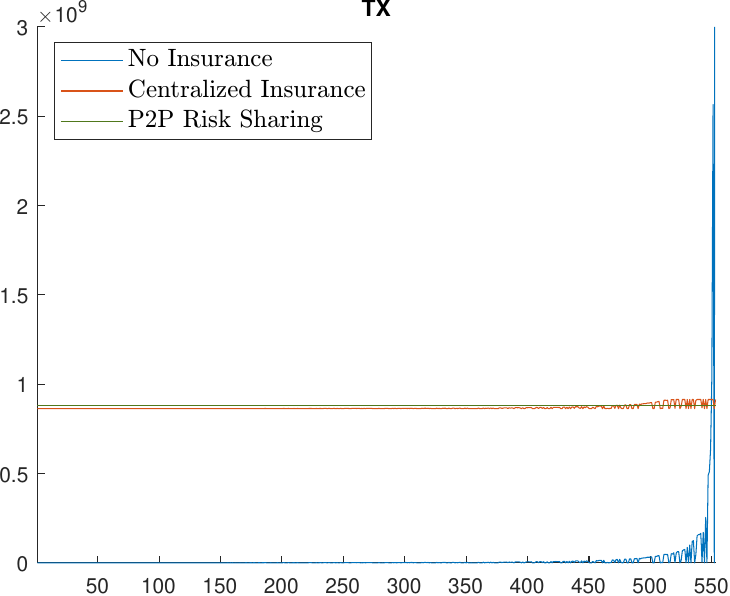}
	\caption{Retained Loss, CA/NY/TX.}
    \label{fig:scheme_comparison}
\end{figure}

A comparison between the centralized and decentralized insurance schemes is shown in Figure \ref{fig:scheme_comparison}, which also includes the scenario of no insurance. The vertical axis represents the retained monthly loss in dollars, and the horizontal axis represents the 553 states of the world $\omega\in\Omega$, which are sorted by the value of $S(\omega)$, the realized aggregate loss in state $\omega \in \Omega$. For ease of comparison, the plots have been normalized so that each agent is indifferent between each retention structure. In the case of centralized insurance, we see that each agent pays a premium to insure their risk beyond a deductible. This results in significantly less variance in retained risk, at a cost of a monthly premium. On the other hand, the structure of the Pareto-optimal risk-sharing arrangement is different for each agent. The plots show that California's retention is similar to that under centralized insurance, in the sense that variation is reduced and extreme losses are insured. However, extreme losses are transferred to New York, who accepts a payment of a premium as compensation. Since there is no exogenous insurer in the decentralized market, all flood risk must be retained among the three agents in the market. The plot in Figure \ref{fig:scheme_comparison} shows that New York is the agent that accepts the extreme tail risk in this example. Finally, Texas elects to pay a premium to transfer all the flood risk to the other two agents, and so its retention structure is a flat horizontal line. This suggests that Texas is the most risk-averse in this scenario.

By Theorem \ref{thm:po_drm_characterization}, both the aggregate welfare gain and the average welfare gain are maximized by Pareto-optimal allocations. Hence, evaluating the effectiveness of the peer-to-peer risk sharing scheme is possible by using the average welfare gain as a metric. Recall that the welfare gain for agent $i$ is defined as the improvement in the risk-sharing scheme over the status quo:
    \[
        \rho_i(X_i)-\rho_i(Y_i+\pi_i)\,.
    \]
The average welfare gain for the peer-to-peer risk sharing scheme is the arithmetic average of the welfare gains of the three agents in the market. In this example, the average welfare gain is $1.9512\times10^8$. While this is a significant gain, this is not quite as large as the average gain of $1.9946\times10^8$ that is possible from a centralized insurance market with the monopolistic insurer. Hence, we see the peer-to-peer risk sharing scheme as a compromise compared to the centralized structure. While a larger welfare gain is ultimately possible in the centralized setting, this carries the risk that the welfare gain is absorbed by only the monopolistic insurer through aggressive pricing schemes; again, see Ghossoub and Zhu \cite{GhossoubZhu2024stackelberg}. On the other hand, the decentralized market admits no such risk, since the welfare gain is necessarily distributed among the agents participating in the risk-sharing pool. However, the trade-off is that the maximum possible average welfare gain is potentially lower in the peer-to-peer market.

In Table \ref{table:varying_power}, we show the average welfare gain per agent in both centralized and decentralized markets for varying values of the parameter $\gamma_3$ in the distortion function of the third agent $T_3$. It can be verified that power distortion functions have a constant relative index of probabilistic risk aversion. As $\gamma_3$ increases, the $RPRA$ of the third agent decreases, and the hence the third agent eventually prefers to retain more of the flood risk in the market. In the extreme case, we see that the peer-to-peer market can actually dominate that of the centralized market in terms of possible welfare gain, since the decentralized scheme allows for the third agent to insure others' risk whereas the centralized scheme does not. Hence, it is possible for the decentralized market's average welfare gain to exceed that of the centralized market. However, in practical settings, we would expect that the policyholders are more risk averse than the central authority, which would lead to an advantage for the centralized insurance market in general.

\begin{table}[htbp!]
    \centering
    \begin{tabular}{cccc}
        \hline
        $RPRA_3$&Centralized Market&Decentralized Market&Percent Decrease\\
        \hline
        \hline
        $0.60$&$1.9946\times10^8$&$1.9512\times10^8$&$2.1750\%$\\
        $0.55$&$1.4460\times10^8$&$1.2112\times10^8$&$16.2383\%$\\
        $0.50$&$1.0415\times10^8$&$6.6055\times10^7$&$36.5774\%$\\        $0.45$&$7.4503\times10^7$&$4.8805\times10^7$&$34.4927\%$\\
        $0.40$&$5.3190\times10^7$&$5.0507\times10^7$&$5.0434\%$\\        $0.35$&$3.8565\times10^7$&$5.2590\times10^7$&$-36.3683\%$\\
        $0.30$&$3.0251\times10^7$&$5.4502\times10^7$&$-80.1628\%$\\
        \hline
    \end{tabular}
    \caption{Average Welfare Gain with Varying Parameter $\gamma_3$.\vspace{0.4cm}}
    \label{table:varying_power}
\end{table}

\section{Conclusion}
\label{sec:conclusion}

In this paper, we obtain a novel characterization of Pareto-optimal allocations in risk sharing markets when agents use robust distortion risk measures. As a special case of our result, we recover the characterization of optima under distortion risk measures, and we provide some general results for risk-sharing using common families of distortion functions. We find that the form of allocations depends primarily on the risk attitude of each agent to tail events, which we quantify through the probabilistic risk aversion index.

As an application of our results, we reexamine the setting of flood risk in the United States from the perspective of a decentralized market. While this market currently operates under a centralized structure with the federal government as the central authority, this comes at a risk associated with Stackelberg equilibria, namely that the central insurer has an incentive to raise prices and eliminate any welfare gain for the policyholders. We therefore argue that the decentralized structure is an alternative that avoids the Stackelberg situation, and we characterize optimal decentralized allocations using historical data. Our results show that there is less welfare gain possible in the decentralized market compared to the centralized one, which we interpret as a necessary compromise.

\newpage
\bibliographystyle{plain}
\bibliography{biblio}
\vspace{0.4cm}

\end{document}